\documentclass[10pt,journal,compsoc]{IEEEtran}
\usepackage{subfigure}
\usepackage[final]{graphicx}
\usepackage{amsmath}
\usepackage[utf8]{inputenc}
\usepackage[english]{babel}
\usepackage{amsthm}
\usepackage{bbm}
\usepackage{amssymb}
\usepackage{xcolor}

  \usepackage{cite}

\newtheorem{thm}{Theorem}
\newtheorem{defi}{Definition}
\newtheorem{prop}{Proposition}
\newtheorem{lem}{Lemma}
\newtheorem{cor}{Corollary}
\newtheorem{rmk}{Remark}

\begin{document}
%
\title{Local Information Privacy and Its Application to Privacy-Preserving Data Aggregation}
%
%
%
%

\author{Bo Jiang,~\IEEEmembership{Student Member,~IEEE,}
        Ming Li,~\IEEEmembership{Senior Member,IEEE,}
        and~Ravi~Tandon,~\IEEEmembership{Senior Member,IEEE}
\IEEEcompsocitemizethanks{\IEEEcompsocthanksitem Bo Jiang, Ming Li and Ravi Tandon are with the Department
of Electrical and Computer Engineering, University of Arizona, Tucson,
AZ, 85721.\protect\\
E-mail: bjiang@email.arizona.edu, lim@email.arizona.edu, tandonr@email.arizona.edu}}

%
%

\markboth{Journal of \LaTeX\ Class Files,~Vol.~14, No.~8, August~2015}%
{Shell \MakeLowercase{\textit{et al.}}: Bare Demo of IEEEtran.cls for Computer Society Journals}
%



\IEEEtitleabstractindextext{%
\begin{abstract}
In this paper, we propose local information privacy (LIP), and design LIP based mechanisms for statistical aggregation while protecting users' privacy without relying on a trusted third party. The concept of context-awareness is incorporated in LIP, which can be viewed as exploiting of data prior (both in privatizing and post-processing) to enhance data utility. 
We present an optimization framework to minimize the mean square error of data aggregation while protecting the privacy of each user's input data or a correlated latent variable by satisfying LIP constraints. Then, we study optimal mechanisms under different scenarios considering the prior uncertainty and correlation with a latent variable. Three types of mechanisms are studied in this paper, including randomized response (RR), unary encoding (UE), and local hashing (LH), and we derive closed-form solutions for the optimal perturbation parameters that are prior-dependent. We compare LIP based mechanisms with those based on LDP, and theoretically show that the former achieve enhanced utility. We then study two applications: (weighted) summation and histogram estimation, and show how proposed mechanisms can be applied to each application. Finally, we validate our analysis by simulations using both synthetic and real-world data. Results show the impact on data utility by different prior distributions, correlations, and input domain sizes. Results also show that our LIP-based mechanisms provide better utility-privacy tradeoffs than LDP-based ones.
\end{abstract}

\begin{IEEEkeywords}
privacy-preserving data aggregation, local information privacy, information-theoretic privacy
\end{IEEEkeywords}}

\maketitle

\IEEEdisplaynontitleabstractindextext

%
\IEEEpeerreviewmaketitle

\IEEEraisesectionheading{\section{Introduction}\label{sec:introduction}}

%
%
%
%
\IEEEPARstart{P}{rivacy} issues are crucial in this big data era, as users' data are collected both intentionally or unintentionally by an increasing number of private or public organizations. Most of the collected data is used for ensuring high quality of service, but may also put one's sensitive information at potential risk. For instance, when people are rating movies, their preferences may be leaked; when users are searching for a parking spot nearby using a smartphone, their real locations are uploaded and prone to leakage. Besides the cases where collected data itself is sensitive and causes privacy leakage, non-sensitive data release may also enable malicious inference on one's private attributes: whenever there is a correlation between the collected data and people's private latent attribute, directly releasing it causes privacy leakage. For instance, heartbeat data collected by smartwatch may potentially reveal one's heart disease \cite{Smartwatch}; One can easily infer a target user's home or work location by tracking his daily location trace\cite{10.1007/978-3-540-72037-9_8}; Smart meters can reveal the activities of people inside a home by tracking their electricity, gas, or water usage frequently over time\cite{Smartmeter}. 
It is, therefore, desirable to design privacy-preserving mechanisms providing privacy guarantees without affecting data utility.

Traditional privacy notions such as  $k$-anonymity \cite{SS98} do not provide rigorous privacy guarantees and are prone to various attacks. Nowadays, Differential Privacy (DP) ~\cite{Dwork20061} has become the \textit{de facto} standard for ensuring data privacy in the database community~\cite{Dwork2006} and has been adopted by the U.S. Census in 2020\cite{inproceedingsCensus}. The definition of DP assures that each user's data has minimal influence on the output of statistical queries on a database. In the classical DP setting, a trusted server is assumed to hold all users' data and provide noisy answers to queries. However, organizations or companies collecting users' data may not be trustworthy, and the data storage system may not be secure. As a result, recently, local privacy protection mechanisms have gained attention as the local setting allows data aggregation while protecting each user's data without relying on a trusted third party.

\subsection{Local Privacy Notions}

 In local privacy-preserving data release, individuals perturb their data locally before uploading it. Organizations that want to take advantage of users' data then aggregate over the collected data. { The earliest such mechanism is randomized response (RR)\cite{randomresponse}, which randomly perturbs each user's data. However, the original RR does not have formal privacy guarantees.
Later, Local Differential Privacy (LDP) was proposed as a local variant of DP that bounds the privacy leakage in the local setting\cite{Freudiger:2011:EPR:2186383.2186387}. Many schemes were proposed under the notion of LDP. For example, \cite{Extreme_ldp,rr_ldp,ldp_lalitha}, and Google's RAPPOR \cite{Rappor}.
LDP based data aggregation mechanisms have already been deployed in the real-world.} For example, in June 2016, Apple announced that it would deploy LDP-based mechanisms for data collection \cite{Apple}. However, Tang \textit{et al}. show that although Apple’s deployment ensures that the privacy budget \footnote{The parameters, $\epsilon\ge{0}$, measures the privacy level. A smaller $\epsilon$ corresponds to a higher privacy level.} of each datum submitted to its servers is $1$ or $2$, the overall privacy budget permitted by the system can be as high as $16$. 
Wang \textit{et al}. proposed a variety of LDP protocols for frequency estimation \cite{Tianhao} and compared their performance with Google's RAPPOR. However, for a given reasonable privacy budget, these protocols provide limited utility. Intuitively, compared with the central DP model, it is more challenging to achieve a good utility-privacy tradeoff in the local setting. The main reasons are: (1) LDP requires introducing noise at a significantly higher level than required in the central setting. That is, for a summation/count query, with additive noise privacy-preserving mechanism, a lower bound of noise magnitude of $\Omega(\sqrt{N})$ is required for LDP in order to defend against potential coalitions of compromised users, where $N$ is the number of users. In contrast, only $O(1)$ is required for central DP \cite{Lowerbound}. 
(2)
LDP does not assume a neighborhood constraint on input data, for data with large domain, LDP leads to a significantly reduced utility \cite{Raef}.
\begin{figure*}[htp]
\centering
\includegraphics[width=14cm]{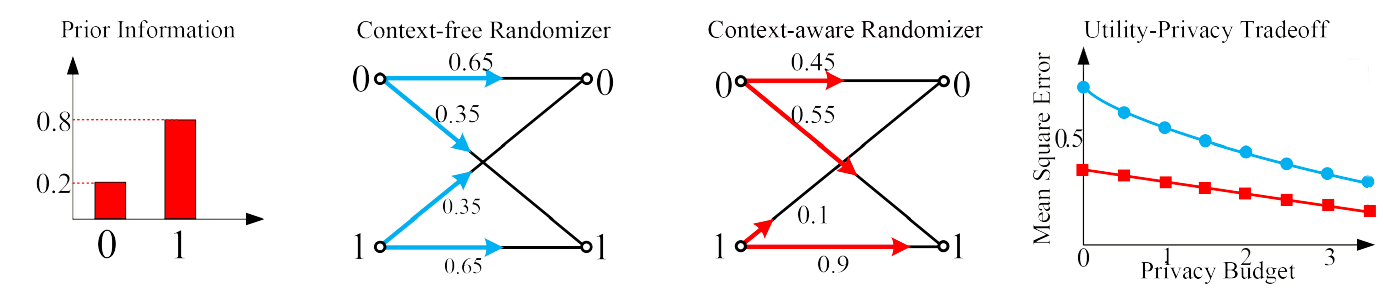}
\centering
\caption{LIP increases utility by explicitly designing perturbation parameters according to prior knowledge.}
\label{fig:example_channel}
\vspace{-13pt}
\end{figure*}

In general, both local and central DP provide strong context-free theoretical guarantees against worst-case adversaries \cite{context}. Context-free means the adversary can possess arbitrary background knowledge of a user's data (except her specific input instance). In other words, the definition of (L)DP is too strong
and regardless of scenarios where the particular context or prior knowledge of the
data is available. Such scenarios exist in many applications. For instance, in Internet of Things (IoT), the prior distribution of context
related to sensor data plays a critical role in distributed data transmission and computation
\cite{IOT}. Another example is location-based services: people have a higher likelihood to be at some locations than others; such as in Paris, people are more likely closer to Eiffel
tower than a coffee shop nearby\cite{Geo}. In mobile-health data collection, background knowledge such as the likelihood of people having certain diseases is available through previously published medical studies\cite{bound}. 
When background information is available, (L)DP fails to capture the explicit
privacy leakage of users or the information gain at the adversary. On the other hand, for a given utility, (L)DP may not always be feasible depending on the privacy budget\cite{freelunch}. Although approximated $(\epsilon,\delta)$-(L)DP is introduced \cite{TCS-042} to realize an achievable mechanism, the non-negative addend $\delta$ could be large enough (close to 1) to provide limited privacy guarantee.

\subsection{Relaxing Local Differential Privacy}

There is a trend among the privacy research community that leverages the background knowledge to relax the definition of DP, and the utility can be increased by explicitly modeling the adversary's knowledge. Privacy notions that consider such prior knowledge are denoted as ``context-aware" privacy notions. For context-aware privacy notions, besides the privacy budget $\epsilon$, the amount of required noise also depends on the prior distribution of the data:  context-dependent privacy mechanisms add noise selectively according to the data prior when most needed so that utility can be enhanced. For example, less noise is required to perturb for data with higher certainty \cite{context,relation}. 
In general the existing context-aware privacy definitions fall into two categories based on either average-case or worst-case guarantees. All information-theoretic privacy notions belong to the former class \cite{ITP1,7498650,aggMIP}. The latter includes Pufferfish \cite{kifer2012rigorous}, Bayes DP\cite{Yang:2015:BDP:2723372.2747643}, Membership privacy\cite{Li:2013:MPU:2508859.2516686}, etc. 
Average-case notions are generally weaker than the latter since they cannot bound the leakage for all the input and output pairs, which may not be easily adopted by the privacy-sensitive users. On the other hand, existing context-aware worst-case privacy notions like Pufferfish and Bayesian DP still follow the same structure of (L)DP -- the maximum ratio between two likelihoods of a certain output given different input data. Since the relationship with prior distribution is not directly captured in the definition, this makes context-aware privacy mechanism design challenging (either high complexity or not easily composable).

\subsection{Local Information Privacy}

In this paper, we make use of the maximum ratio of posterior to prior to capture information leakage in the local setting, denote as local information privacy (LIP). Originally,
information privacy (IP) was proposed in a central setting by Calmon et. al.
\cite{Centrlized_IP}, which requires a trusted curator. 
The main reason that prohibits Centralized IP from being adopted in practice is that the distribution of all users' data is too complex to express or capture, especially for a large-size dataset. In contrast, LIP requires only the prior distribution of one particular user's data, which can be obtained through many approaches in practice.

An illustrative example of why context-aware privacy notions result in increased utility is shown in Fig. \ref{fig:example_channel}, which shows the perturbation mechanisms of context-free (LDP) and context-aware (LIP) notions and the comparison of the mean square errors when collecting private binary data with specific prior. We illustrate the optimal perturbation probabilities for the same privacy budget (epsilon=0.6) under both LDP and LIP privacy notions. Observe that the perturbation channel of LDP is symmetric, while LIP designs perturbation parameters according to the prior knowledge. When the data value is quite certain, it has a smaller probability of flipping the value to increase utility. While when the data takes a value that has a small probability of happening, the mechanism also protects its privacy by a large perturbation probability (a large amount of additive noise). In this example, the probability of flipping the data value through the LDP mechanism is $0.35$ in contrast to $0.2\times{0.55}+0.8\times 0.1=0.19$ of the LIP based mechanism. As a result, LIP leads to an enhanced utility than LDP.

\subsection{Related Work}
In the original paper on differential privacy, Dwork et al. \cite{Dwork2008} defined a notion of ``semantic" privacy that involves comparing the prior and posterior distributions of the database or a user's participation. Since then, similar privacy notions have been investigated in the central setting. Such as, in \cite{Nabar2010}, $\epsilon$-Semantic Privacy is studied, which captures the additional information caused by releasing a contingency table. The privacy is measured by the absolute distance between the prior to posterior ratio and 1, and it allows the ratio to scale linearly with $\epsilon$. In \cite{Kasiviswanathan_Smith_2014}, semantic privacy is redefined by capturing the statistical difference between two posterior beliefs at the adversary. The posterior probabilities are calculated by priors of two neighboring datasets and the same output of the mechanism. In \cite{Rastogi}, privacy is measured by the prior to posterior ratio at the adversary that one user's tuple belongs to a collection of records.  In \cite{Li:2013:MPU:2508859.2516686}, more specifically, $\epsilon$-Membership privacy measures the adversary's prior and posterior beliefs on whether the tuple of the target user belongs to the dataset. However, the privacy notions described above consider the central setting, and the input of the mechanism is a dataset rather than each individual's data. This makes it inconvenient to adapt them into the local setting. To the best of our knowledge, the prior to posterior structure has not yet been thoroughly explored in privacy definitions for the local setting, where each individual releases a privatized answer to an untrusted third party directly. 

To avoid explicitly modeling the adversary's background knowledge, a more robust and practical way to define privacy is to relax the exact prior assumption.  In \cite{bound}, bounded prior differential privacy is studied, which assumes that the real prior distribution comes from a bounded set of probability simplex. In \cite{Rastogi}, it is assumed that the adversary's belief on the targeted individual's membership is upper bounded. Pufferfish privacy \cite{kifer2012rigorous} also assumes bounded knowledge of the adversary. The knowledge is captured by set $\mathcal{P}$, which contains all plausible evolution scenarios of the hidden secret and the input data. As a result, by adjusting the size of $\mathcal{P}$, Pufferfish can be viewed as a generalization of DP while accounting for prior knowledge. In this paper, we also define a bounded set of priors to avoid modeling the adversary's knowledge explicitly.

On the other hand, in many applications, the user's secret information to be protected is different from but correlated with the data being collected. 
To this end, the privacy notions leveraging latent variable like Pufferfish enables a variety of definitions of data utility, such as principal inertia components \cite{DBLP:journals/corr/CalmonMMVCD17}, data pattern\cite{7930028}, distribution estimation\cite{45810}, \textit{etc}. However, one of the drawbacks of Pufferfish privacy is the difficulty of mechanism design. Recently, in \cite{DBLP:journals/corr/WangSC16}, Wang \textit{et al}. designed a Wasserstein Mechanism, which achieves Pufferfish privacy, but it is computationally inefficient, and the mechanism they proposed is approximated. In this work, we combine the bounded prior set and latent variable into the prior to posterior structure, and we show that LIP only assumes the adversary has access to the statistic of the input data, and the correlation with the latent variable, but not the distribution of latent variables. 

To derive the utility-privacy tradeoff, there's a line of work that formulating optimization problems to maximize the utility while subject to certain privacy constraints or doing conversely\cite{Tianhao,relation,Rappor,DBLP:journals/corr/abs-1807-11317,Jian1805:Context}. Firstly, most of them define utility for some specific applications, such as frequency estimation, itemset aggregation, statistic estimation, etc. In this paper, we consider a general type of utility defined by the mean square error of a function of the input and output. We showed that by instantiating it with different functions, the proposed mechanisms could be applied to multiple real-world applications. Secondly, only a few works above provide closed-form optimal solutions for mechanism design. In \cite{Tianhao}, optimization problems are formulated to increase the accuracy in frequency estimation, and different protocols are studied under LDP. However, the utility provided by various mechanisms is limited because no prior information about the data is incorporated into the mechanism. In \cite{DBLP:journals/corr/abs-1807-11317}, the utility optimized LDP mechanism is proposed, which is shown to achieve better utility by exploiting different data input's sensitivity, which is a different type of context than priors. 

\subsection{Main Contributions}
The  main contributions of this paper are listed as follow:

(1)
We propose Local Information Privacy (LIP) for local data release (without a trusted third party), which relaxes the notion of LDP by incorporating prior knowledge and introducing latent variables. We formally derive the relationships between existing privacy definitions and LIP.

(2) We apply LIP to privacy-preserving data aggregation: we present a general framework to estimate a function of the collected data and minimize the mean squared error of the estimation while protecting each individual's privacy by satisfying LIP constraints. We consider three perturbation mechanisms. One can be viewed as a general form of the RR; the other two incorporate unary encoding and local hashing. We derive the optimal mechanisms for different scenarios on prior uncertainty and correlation between input data and latent secret.

(3) We consider two real-world applications in this paper, including weighted summation and histogram estimation. We demonstrate that considering prior knowledge helps the curator design an unbiased estimator, which significantly improves data utility by post-processing; On the other hand, for the users, we show how proposed mechanisms can be applied to these two applications. Compared with LDP based mechanisms, we show that LIP based mechanisms provide enhanced utility.

(4) We validate our analysis by simulations on both synthetic and real-world datasets (Karosak, a website-click stream data set, and Adult, a survey of census income). We illustrate the impact of data correlation, input data domain, and prior uncertainty on data utility provided by different mechanisms. When compared to LDP based mechanisms, LIP based mechanisms always provide better utility. For input data with a large domain, encoding methods could potentially increase utility than compared to RR. 

\subsection{Paper Organization}
The remainder of the paper is organized as follows: In Section~\ref{sec:privacy_model}, we introduce the proposed LIP notion and its relationship with other existing privacy notions. In Section~\ref{sec:model}, we introduce the system model and problem formulation. In Section~\ref{tradeoff}, we derive the utility-privacy tradeoff, including model with a fixed prior, model with an uncertain prior. Under each model, encoding based mechanisms are studied. Then, we compare with LDP based model. Finally, we discuss the applications of these models, including weighted summation and histogram estimation. In Section~\ref{sec:sim}, we present the simulation results and compare the utility-privacy tradeoffs provided by different mechanisms under different data domain, data prior, data correlations with different datasets. In Section~\ref{sec:con}, we offer concluding remarks.

\section{Privacy Definitions and Relationships}\label{sec:privacy_model}
In this Section, we first recap several existing privacy notions in the local setting. We then introduce LIP and study its relationships with other notions. In this paper, we focus on discrete-valued data.

\subsection{Privacy Definitions}
Consider a privacy-protection mechanism $\mathcal{M}$ takes input data $X$ and outputs a perturbed version of $Y$. It is assumed that $X$ takes value from a discrete domain $\mathcal{X}$ with the prior distribution of $\theta_X\in\mathcal{P}_{\mathcal{X}}$, where $\mathcal{P}_{\mathcal{X}}$ is the set containing all possible prior distributions on $\mathcal{X}$. In the latent variable setting, denote $G$, which takes value from $\mathcal{G}$ as the hidden secret that is correlated with $X$. Denote $\theta_{XG}\in\mathcal{P}_{\mathcal{X}\mathcal{G}}$ as the joint distribution of $X$ and $G$. Denote $\mathcal{Y}=\text{Range}(\mathcal{M})$ as the domain of $Y$.


The context-free LDP definition states that any two inputs from the data domain $\mathcal{X}$ result in the same output with similar probabilities.

\begin{defi}($\epsilon$-Local Differential Privacy (LDP))\cite{Extreme_ldp}\label{def:LDP}
 $\mathcal{M}$  satisfies  $\epsilon$-LDP for some $\epsilon\in{\mathbf{R}^+}$, if $\forall{x, x'\in{\mathcal{X}}}$ and $\forall{y\in{\mathcal{Y}}}$: 
\begin{equation}
    \frac{Pr(Y=y|X=x)}{Pr(Y=y|X=x')}\le{e^{\epsilon}}.
\end{equation}
\end{defi}
LDP provides strong context-free privacy protection, since it provides indistinguishability of input's data-value regardless of the data prior distribution. Context-free notions typically suffer poor utility-privacy tradeoff. We next introduce context-aware privacy definitions. 

Maximal Information Leakage captures the adversary's ability without assuming a particular accessible prior.

\begin{defi}($\epsilon$-Maximal Information Leakage (MIL))\cite{7460507} The maximal information leakage of $\mathcal{M}$ is defined as: 
\begin{align}
 \mathcal{L}(X\to{Y})=\log \sum_{y\in{\mathcal{Y}}}\max_{x\in{\mathcal{X}}}Pr(Y=y|X=x),
\end{align}
and $\mathcal{M}$ satisfies Maximal Information Leakage privacy if for some $\epsilon\in{\mathbf{R}^+}$: $\mathcal{L}(X\to{Y})\le{\epsilon}$.
\end{defi}
MIL captures the average likelihood probability over all possible $y\in{\mathcal{Y}}$ given the corresponding value of $x$ that maximizes this probability. However, MIL does not provide pairwise protection over all possible values of $x$ and $y$ and hence is relatively weak. 


Mutual information privacy measures the average information leakage of $X$ contained in $Y$: 

\begin{defi}($\epsilon$-Mutual Information Privacy (MIP))\cite{7498650}
 $\mathcal{M}$  satisfies $\epsilon$-MIP for some $\epsilon\in{\mathbf{R}^+}$, if the mutual information between $X$ and $Y$ satisfies $I(X;Y)\le{\epsilon}$, where $I(X;Y)$ is: 
\begin{small}
\begin{equation}\label{mutualinfo}
    \sum_{x\in{\mathcal{X}},y\in{\mathcal{Y}}}Pr(X=x, Y=y)\log\frac{Pr(X=x,Y=y)}{Pr(X=x)Pr(Y=y)}.
\end{equation}
\end{small}
\end{defi}

\noindent Although MIP is context-aware, it provides relatively weak privacy protection since it only bounds the average information leakage over all possible $x$ and $y$ in the domain. 

Another context-aware privacy notion that provides pairwise protection over each possible values of $x$ and $y$ is differential identifiability.

\begin{defi}($\epsilon$-Differential Identifiability (DI)) \cite{Lee:2012:DI:2339530.2339695}
$\mathcal{M}$ satisfies  $\epsilon$-DI for some $\epsilon\in{\mathbf{R}^+}$, if $\forall{x, x'\in{\mathcal{X}}}$ and $\forall{y\in{\mathcal{Y}}}$: 
\begin{equation}
    \frac{Pr(X=x|Y=y)}{Pr(X=x'|Y=y)}\le{e^{\epsilon}}.
\end{equation}
\end{defi}
\noindent The operational meaning of DI is, given the output $y$, the adversary cannot tell whether the original data(set) is $x$ or $x'$. DI can be directly adapted in the local setting, and is context-aware due to the dependence on the data prior:
\begin{equation*}
    \frac{Pr(Y=y|X=x)Pr(X=x)}{Pr(Y=y|X=x')Pr(X=x')}\le{e^{\epsilon}}.
\end{equation*}
One major drawback of DI is the difficulty of designing practical mechanisms, as DI measures the ratio of posteriors, which means the likelihood ratio (perturbation parameters) of any two different inputs is dependent on the prior ratio. For example, if $\frac{Pr(X=x)}{Pr(X=x')}$ is small, DI requires $\frac{Pr(Y=y|X=x')}{Pr(Y=y|X=x)}$ to be large for all $y\in{\mathcal{Y}}$. However, we know that $\sum_{y\in{\mathcal{Y}}}Pr(Y=y|X=x')=\sum_{y\in{\mathcal{Y}}}Pr(Y=y|X=x)=1$.

Pufferfish privacy is originally proposed in the central setting\cite{kifer2012rigorous}, and here we adapt it into the local setting where $X$ and $Y$ stand for user's input and output data, respectively.
\begin{defi}[Local Pufferfish Privacy]

Given a set of potential secrets $\mathcal{G}$, a set of discriminative pairs $\mathcal{G}_{pairs}$, a set of data evolution scenarios $\mathcal{P}_{\mathcal{X}\mathcal{G}}$,  $\mathcal{M}$ satisfies $\epsilon$-Pufferfish ($\mathcal{G}$, $\mathcal{G}_{pairs}$, $\mathcal{P}_{\mathcal{X}\mathcal{G}}$) privacy, for some $\epsilon\in{\mathbf{R}^+}$ if
\begin{itemize}
    \item for all possible inputs $x\in{\mathcal{X}}$, $y \in \mathcal{Y}$,
    \item for all pairs $(g_i,g_j)\in\mathcal{G}_{pairs}$ of potential secrets,
    \item for all distributions $\theta_{XG}\in \mathcal{P}_{\mathcal{X}\mathcal{G}}$ that $Pr(g_i|\theta_{XG})\neq{0}$ and $Pr(g_j|\theta_{XG})\neq{0}$,
\end{itemize}
the following holds:
\begin{equation}
    e^{-\epsilon}\le\frac{Pr(\mathcal{M}(x)=y|\theta_{XG},g_i)}{Pr(\mathcal{M}(x)=y|\theta_{XG},g_j)}\le{e^{\epsilon}}.
\end{equation}
\end{defi}
\noindent Note that, when the set $\mathcal{P}_{\mathcal{X}\mathcal{G}}$ spans all possible joint distributions including the case when $X=G$. Then for such a special case, Local Pufferfish becomes equivalent to LDP.

Motivated by central information privacy\cite{Centrlized_IP}, to provide a pairwise constraint on the information leakage of secret $G$ through $Y$ in the local setting, we consider a bound on the ratio between the prior and posterior, which leads to the notion of local information privacy. Denote $\theta_X$ as the prior of input data $X$, $T_{GX}$ as the conditional probability of $Pr(X=x|G=g)$. Denote $\theta_{XG}$ as a fixed data evolution scenario: $\theta_{XG}=\{\theta_X, T_{GX}\}$. The definition of Local Information Privacy is defined as:
\begin{figure*}[t]
\centering 
\subfigure[A summary of different privacy notions]
{ \includegraphics[width=0.38\textwidth]{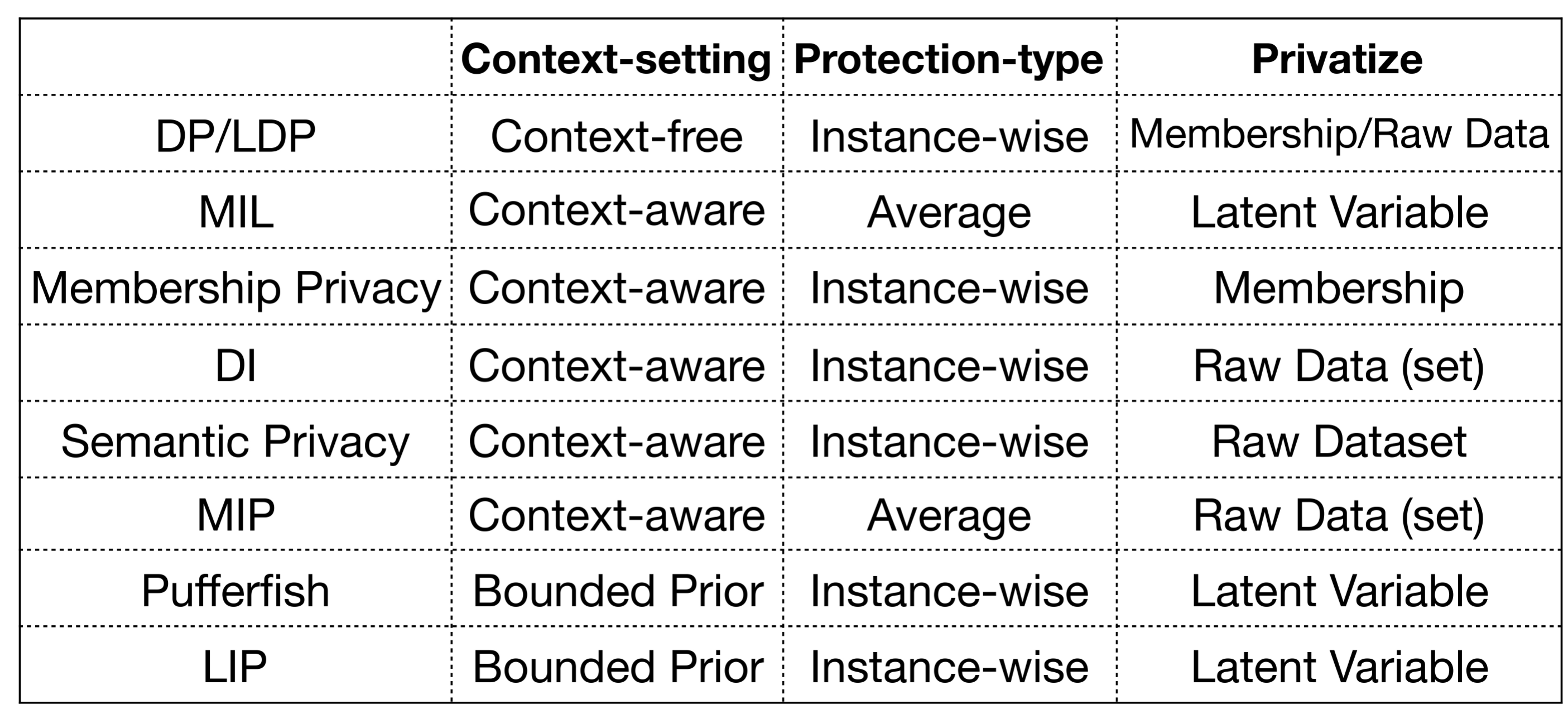}
\label{notions} }
\subfigure[{Relationship between LIP and other privacy notions.}] 
{ \includegraphics[width=0.58\textwidth]{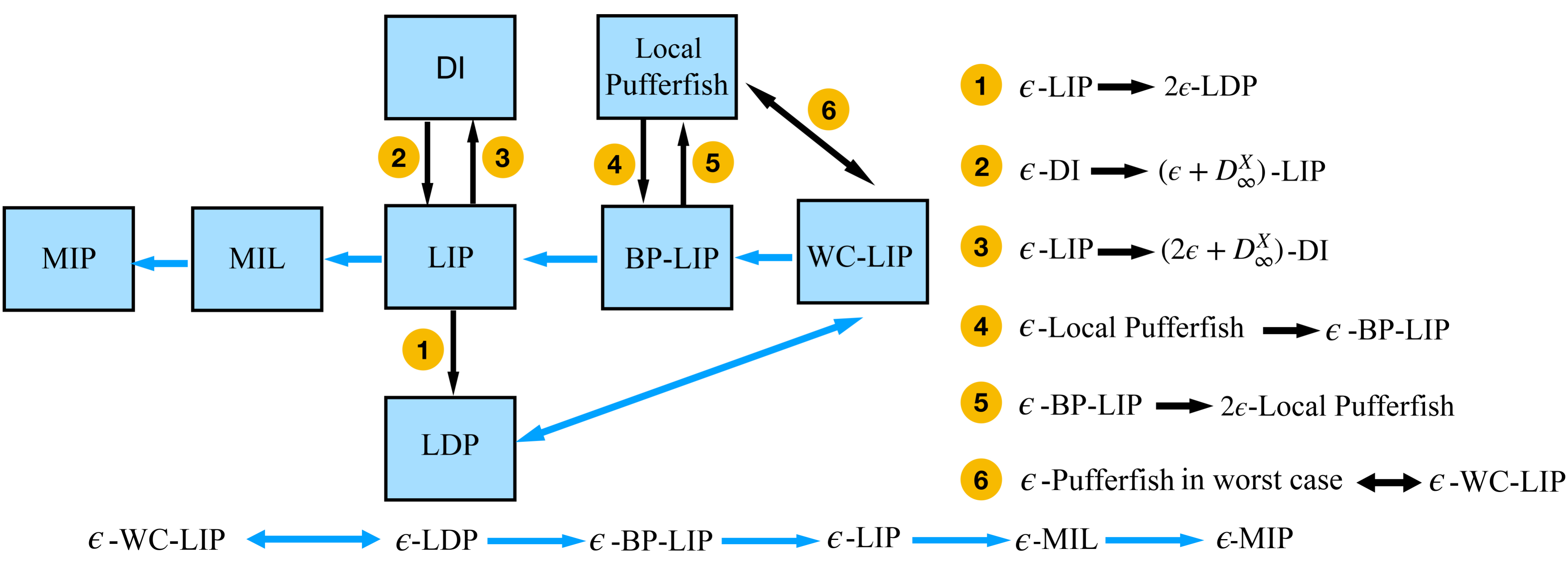}
\label{fig:relationship} }
\caption{A summary of different privacy notions and the relationships among them.}
\vspace{-10pt}
\label{fig:notions}
\end{figure*}


\begin{defi}($\epsilon$-Local Information Privacy (LIP))\label{def:LIP}
Given a set of potential secrets $\mathcal{G}$, given a set of data evolution scenarios $\mathcal{P}_{\mathcal{X}\mathcal{G}}$, $\mathcal{M}$ satisfies $\epsilon$-LIP for some $\epsilon\in{\mathbf{R}^+}$, if $\forall{g\in{\mathcal{G}}}$, $\forall{\theta_{XG}\in\mathcal{P}_{\mathcal{X}\mathcal{G}}}$ and $\forall{y\in{\mathcal{Y}}}$:
\begin{equation}\label{eq1}
    e^{-\epsilon}\le{\frac{Pr(G=g|\theta_{XG})}{Pr(G=g|Y=y,\theta_{XG})}}\le{e^{\epsilon}}.
\end{equation}
There are three cases regarding the range of $\mathcal{P}_{\mathcal{X}\mathcal{G}}$:
\begin{itemize}
    \item When $\mathcal{P}_{\mathcal{X}\mathcal{G}}$ includes one given prior distribution, LIP becomes LIP for fixed prior $\theta_{{X}{G}}$;
    \item When $\mathcal{P}_{\mathcal{X}\mathcal{G}}$ includes all possible priors, LIP becomes Worst-Case-LIP (WC-LIP);
    \item When $\mathcal{P}_{\mathcal{X}\mathcal{G}}$ includes a subset of all possible priors, LIP becomes Bounded-Prior-LIP (BP-LIP).
\end{itemize}
\end{defi}

The operational meaning of LIP is: By observing any output $y$, the change of the belief about the latent variable taking any specific value compared with the prior distribution is not increased or decreased too much. Note that, when $\epsilon$ is small, this ratio is bounded close to 1, which means the output $Y$ is independent of the latent secret $G$.

Note that, LIP assumes the adversary is accessible to the statistic of the input data and the conditional probability of $Pr(X=x|G=g)$, but may not be accessible to the prior of the secret $G$. Such assumption also helps avoid modeling the adversary's ability explicitly. Moreover, it also enables LIP to protect either discrete or continuous-valued secret $G$.

LIP also guarantees that any post-processing on the output cannot further increase privacy leakage.
\begin{lem}
When $G\to X\to{Y}\to{Z}$ forms a Markov chain, if for any $y\in\mathcal{Y}$ and $g\in\mathcal{G}$, $\mathcal{M}$ satisfies $\epsilon$-LIP, then $\mathcal{M}$ also guarantees $\epsilon$-LIP for any $z\in\mathcal{Z}$ and $g\in\mathcal{G}$.
\end{lem}
\begin{proof}
As $Pr(G=g|Z=z)=\sum_{y\in\mathcal{Y}}Pr(G=g|Y=y)Pr(Y=y|Z=z)$, which is bounded between ${\min_{y\in\mathcal{Y}}}Pr(G=g|Y=y)$ and ${\max_{y\in\mathcal{Y}}}Pr(G=g|Y=y)$. Since the ratio of $Pr(G=g)/Pr(G=g|Y=y)$ is bounded by $[e^{-\epsilon},e^{\epsilon}]$ for all $g\in\mathcal{G},y\in\mathcal{Y}$, the ratio of $Pr(G=g)/Pr(G=g|Z=z)$ is also bounded by $[e^{-\epsilon},e^{\epsilon}]$
\end{proof}
Such property enables the data curator to do further data mining, without increasing the privacy leakage. Compared to other context-aware definitions, LIP (including BP-LIP and WC-LIP) models the prior attainability comprehensively, including the scenarios where the prior is uncertain, (WC-LIP can be viewed as context-free). 

 
 
 


\subsection{Relationships with Existing Definitions}

\subsubsection{LIP v.s. LDP}
Since LDP does not assume a latent variable, to make a fair comparison between LIP and LDP, we assume the input $X$ is private, i.e., $G=X$. Then, the following relationship holds between fixed-prior LIP and LDP:  $\epsilon$-LIP implies
 $2\epsilon$-LDP and $\epsilon$-LDP implies $\epsilon$-LIP (proof is shown in\cite{Jian1805:Context}).  This implies that  $\epsilon$-LIP is a more relaxed privacy notion than $\epsilon$-LDP. However, it is stronger than $2\epsilon$-LDP. 
 
 When comparing the relationship between $\epsilon$-WC-LIP and $\epsilon$-LDP, we have $\epsilon$-WC-LIP is equivalent to $\epsilon$-LDP (proof is shown in \cite{8761176}). Intuitively, these two definitions are equivalent because both of them assume worst-case (context-free) priors. Then the relationship between LDP and BP-LIP is straightforward: $\epsilon$-BP-LIP is sandwiched between $\epsilon$-LDP and $\epsilon$-LIP. As a result, LIP, BP-LIP, and WC-LIP can be viewed as context-aware versions of LDP with different assumptions on the data priors. 
 We further compare the utility privacy tradeoff between these two definitions in terms of optimal mechanism design in Sec. \ref{sec:LDP}. 

\subsubsection{LIP v.s. Local Pufferfish}


We next compare LIP (BP-LIP, WC-LIP) with Local Pufferfish privacy according to different scenarios of $\mathcal{P}_{\mathcal{X}\mathcal{G}}$. The results in the next lemma follow from the proof of the relationship between LIP and LDP.
\begin{lem} The relationship between $\epsilon$-LIP and $\epsilon$-Local Pufferfish can be described as follow:
\begin{itemize} 
    \item $\epsilon$-WC-LIP is equivalent to $\epsilon$-Local Pufferfish when $\mathcal{P}_{\mathcal{X}\mathcal{G}}$ includes all possible $\theta_{XG}$;
    \item $\epsilon$-Local Pufferfish implies $\epsilon$-BP-LIP, and $\epsilon$-BP-LIP implies $2\epsilon$-Local Pufferfish when $\mathcal{P}_{\mathcal{X}\mathcal{G}}$ includes a subset of all possible prior distributions of $\theta_{XG}$.
\end{itemize}
\end{lem}
\noindent When $\mathcal{P}_{\mathcal{X}\mathcal{G}}$ includes all possible prior distributions of $X$ and $G$, $\epsilon$-Local Pufferfish considers $X=G$ (where the leakage is maximized), which is equivalent to $\epsilon$-LDP.

In summary, Local Pufferfish relaxes LDP by defining a bounded set of possible prior distributions. Since the structure of the ratio of two likelihoods in the definition of LDP does not allow for the incorporation of prior knowledge, Pufferfish further extends it by a correlated latent variable. This definition is more general in terms of operational meaning than only protecting the input. However, it also comes with difficulties in mechanism design compared to LDP, as the values of $Pr(Y=y|G=g)$ averages over all the likelihood probabilities of $Pr(Y=y|X=x)$, which are the perturbation parameters.

\subsubsection{LIP v.s. Other Privacy Notions}

We next compare the relationship between LIP and MIP, MIL and DI. Since these definitions do not assume latent variable or bounded prior set, we simplify the definition of LIP by: given the prior $\theta_X$, a mechanism $\mathcal{M}$ satisfies $\epsilon$-LIP for some $\epsilon\in{\mathbf{R}^+}$ if $\forall{x\in\mathcal{X}}$, $y\in\mathcal{Y}$:
\begin{equation}
    e^{-\epsilon}\le{\frac{Pr(Y=y)}{Pr(Y=y|X=x)}}\le{e^{\epsilon}}.
\end{equation}

\noindent Then, $\epsilon$-LIP provides stronger privacy guarantee than $\epsilon$-MIP, since $\frac{Pr(X=x,Y=y)}{Pr(X=x)Pr(Y=y)}=\frac{Pr(X=x|Y=y)}{Pr(X=x)}\le{e^{\epsilon}}$. $\epsilon$-LIP also implies $\epsilon$-MIL, as $\max_{x\in{\mathcal{X}}}Pr(Y=y|X=x)\le{Pr(Y=y)e^{\epsilon}}$. Intuitively, among LIP, MIP and MIL, only LIP provides pairwise protection over each possible realization of $x$ and $y$. 
 To compare the relationship between LIP and DI, we first define the maximal ratio of two prior probabilities of $X$ as $D^{{X}}_{\infty}=\max_{x,x'\in{\mathcal{X}}}\log{\frac{Pr(X=x)}{Pr(X=x')}}$, then, the relationship between LIP and DI follows the next lemma with proof provided in Appendix A of the supplementary document.
 \begin{lem}
 The relationship between LIP and DI is:  $\epsilon$-LIP implies $(2\epsilon+D^{{X}}_{\infty})$-DI and $\epsilon$-DI implies $(\epsilon+D^{{X}}_{\infty})$-LIP.
 \end{lem}
The characteristics, relationships, and order among different privacy notions are summarized in Fig. \ref{fig:notions}.
So far, if a mechanism satisfies $\epsilon$-LIP, it implies $\epsilon$-MIP, $\epsilon$-MIL, $2\epsilon$-LDP, $2\epsilon
$-Pufferfish and $(2\epsilon+D^{{X}}_{\infty})$-DI. The main reasons that we choose to study LIP instead of other notions are listed as follows: (1) LIP is more amenable to incorporate prior knowledge to design mechanisms than other context-aware notions. (2) Compared to context-free notions, LIP based mechanisms achieve much higher utility. 

In the following sections, we address how to design LIP based mechanisms according to the prior knowledge, and how LIP based mechanisms improve the utility-privacy tradeoff for different types of applications.

\section{Models  and Problem Formulation}\label{sec:model}
\subsection{System and Threat Models}
\begin{table}[t]\label{tbl1}
\caption{List of symbols}
\scriptsize
\centering
\begin{tabular}{c{l} c{l} c{l} c{l}}
\hline
$\mathcal{R}$ & The universe of raw data values & $R$ & Raw data\\
$\theta$ & Prior distribution & $G$ & Private latent variable\\
$\mathcal{X}$ & The universe of input values & $X$ & Input random variable\\
$T$ & Correlation with latent variable& $\bar{X}$ &  Set of input data\\
$Y$ & Output random variable &  $\bar{Y}$ & Set of output data  \\
$N$ & Total number of users & $\mathcal{M}$ & Privatizing mechanism\\
$\mathbf{q}$ & Set of perturbation parameters&$f(\cdot)$ & Aggregation function\\
$\hat{X}$ & Estimator at the curator& $\hat{\mathbf{S}}$ & Aggregated result\\
$\epsilon$ & Privacy budget & $U$ & Utility measurement\\
$\mathcal{E}$ & Mean square error function & $\mathcal{T}$ & Feasible region of $\mathbf{q}$\\
\hline
\end{tabular}
\vspace{-10pt}
\end{table}

Consider a data aggregation system with  $N$ users and a data curator. Each user possesses discrete-valued data $R_i\in\mathcal{R}$, with the prior distribution of $\theta^i_R$, which can be specified by $P^i_r=Pr(R_i=r)$, where $i\in\{1,2,..,N\}$ is the user index. It is assumed that $R_i$s  are independent of each other (and may have different distributions). Note that, each $R_i$ may be different from but correlated with some private hidden secret $G_i\in\mathcal{G}$. Denote $T^i_{GR}$ as the conditional probability of $R_i$ given $G_i$, and $\theta^i_{RG}=\{\theta^i_R,T^i_{GR}\}$. Denote $\mathcal{P}^i_{\mathcal{R}\mathcal{G}}$ as the bounded set including all possible $\theta^i_{RG}$. {To answer some query, each user locally generates data $X_i$ from $R_i$ by a query-dependent function $f_i$, i.e., $X_i=f_i(R_i)$.  It is assumed that $f_i$ is surjective, i.e., for any $x\in\mathcal{X}$, there is at least one $r\in\mathcal{R}$, s.t., $f_i(r)=x$. Then, the prior distribution of $\theta^i_X$ can be calculated by the prior of $\theta^i_R$ according to the local function $f_i$, and can be specified by $P^i_x=Pr(X_i=x)$. Similarly, the correlation between $X_i$ and $G_i$, $T^i_{GX}$ can be obtained by $T^i_{GR}$ and $f_i$. Denote $\theta^i_{XG}=\{\theta^i_X,T^i_{GX}\}$, and $\mathcal{P}^i_{\mathcal{X}\mathcal{G}}$ as the bounded set including all possible $\theta^i_{XG}$.
To avoid potential privacy leakage, before publishing $X_i$, each user locally perturbs it by a privacy-preserving mechanism $\mathcal{M}_i$. The output is denoted as  $Y_i\in\mathcal{Y}$.} The mechanism maps each possible input to each possible output with a certain probability (perturbation parameter). After receiving each perturbed data, the curator is allowed to further estimate and compute a statistical function of the collected data.  The system model is depicted in Fig. \ref{fig:Model}.

The curator is considered red{honest but curious} due to both internal  and external threats. On one hand, users' private data is profitable,  and companies can be interested in user tracking or selling their data. On the other hand, data breaches happen from time to time due to hacking activities.   The curator aims at
performing accurate estimations using all the information above, but is also interested in inferring each user’s hidden secret $G_i$. Denote the true aggregated result by $S=f(\bar{R})$, where $\bar{R}=\{R_1,R_2,...,R_N\}$. { Later {we} discuss the relationship between $f$ operated at the curator and $f_i$ conducted by each user.} For different applications of data aggregation, the definition of $f(\cdot)$s varies. In this paper, two applications are considered: 
\begin{itemize}
\item Weighted summation: the curator is interested in finding the summation over users' data:  ${S}=\sum^N_{i=1}{(c_iR_i+b_i)}$. When each $c_i=1$ and $b_i=0$, the application is equivalent to a direct summation, which is useful to find the average value;
\item  Histogram estimation: the curator is interested in estimating how many people possess each of the data category in $\mathcal{R}$, or classifying according to users' data value. ${S}$ in histogram is a set of ``categorized" data: $\{S_{1},S_{2},...,S_{|\mathcal{X}|}\}$, such that, $\forall{{k}\in\mathcal{R}}$, $S_{k}=\sum^N_{i=1}\mathbbm{1}_{\{R_i=k\}}$, where $\mathbbm{1}_{\{a=b\}}$ is an indicator function, which is 1 if $a=b$; 0 if $a\neq{b}$. 
\end{itemize}
The curator (adversary) observes all the users' outputs $\bar{Y}=\{Y_1,Y_2,...,Y_N\}$ and tries to obtain an estimation of ${S}$ using estimator $\hat{S}$. 

 \begin{figure}[t]
\centering
\includegraphics[width=9cm]{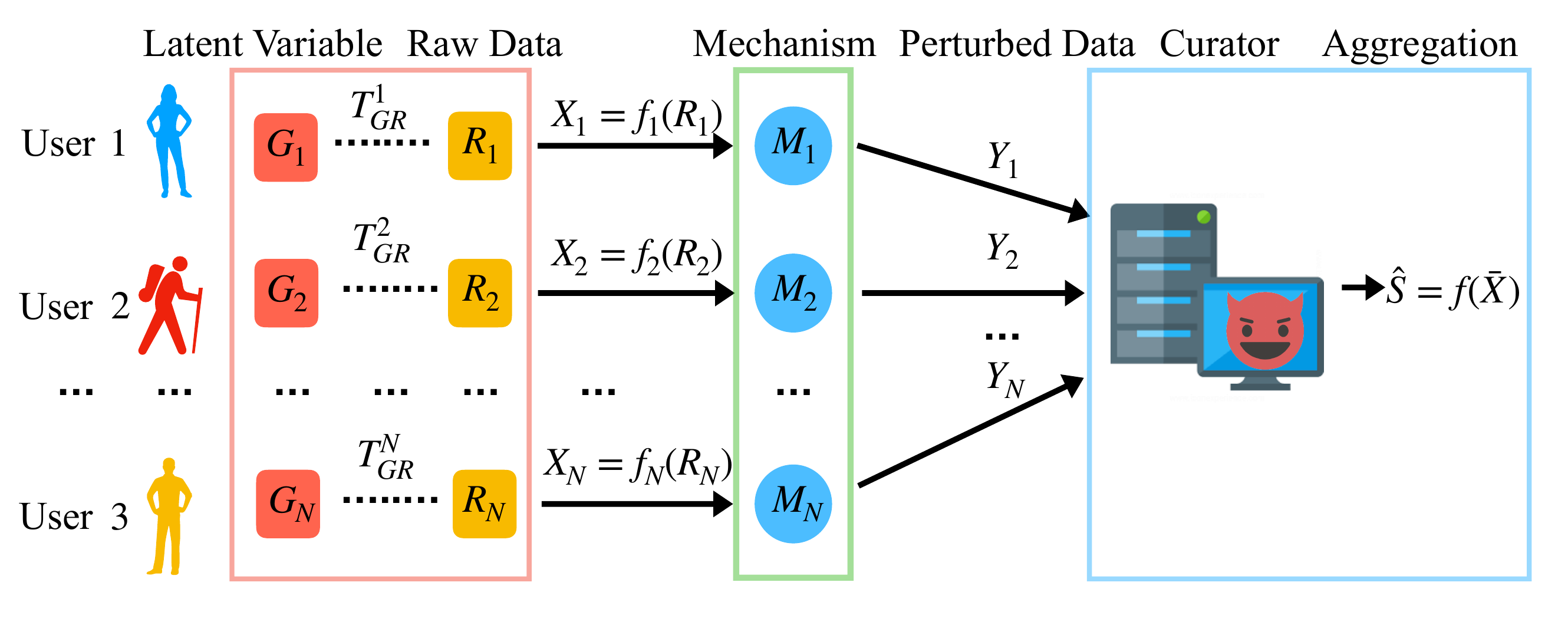}
\setlength{\abovecaptionskip}{-3pt}
\setlength{\belowcaptionskip}{-3pt}
\caption{System Model of Privacy-Preserving Data Aggregation.}
\vspace{-10pt}
\centering
\label{fig:Model}
\end{figure}

 In terms of prior availability, multiple scenarios could arise in practice. For example, both the user and the curator know $\theta^i_R$ exactly, or one party is uncertain about $\theta^i_R$, or they possess different prior knowledge from each other, and one or both of them can be inaccurate.  Within this paper's scope, we assume that the curator always knows the exact $\theta^i_R$ ($\theta^i_{RG}$), and the algorithms/perturbation mechanisms that users deployed to publish their data. 
 In the basic setting, we assume each user also possesses the exact prior (same as the curator). Later we relax it and consider uncertain prior at the user.
All the related symbols are listed in Table 1.

\subsection{General Privacy and Utility Definitions}

The \textit{privacy}   of each user's latent secret is guaranteed by LIP and is parameterized by the privacy budget ($\epsilon$) in Definition (\ref{def:LIP}). The smaller $\epsilon$ is, the stronger privacy guarantee the mechanism provides.  For simplicity, we consider $\epsilon$ to be the same for all the users. However, it is straightforward to extend our model and results to the scenarios where different users are provided by different $\epsilon$s. When the exact prior $\theta^i_{XG}$ is not available for each user, he/she defines $\mathcal{P}^i_{\mathcal{X}\mathcal{G}}$ to be the set of plausible priors including $\theta^i_{XG}$ (users are always allowed to enlarge the size of the $\mathcal{P}^i_{\mathcal{X}\mathcal{G}}$ to include $\theta^i_{XG}$). 
Under LIP, the privacy constraints can be formulated as: $\forall i\in{\{1,...,N\}}$, $\forall g\in{\mathcal{G}}$, $\forall{\theta^i_{XG}\in\mathcal{P}^i_{\mathcal{X}\mathcal{G}}}$ and $y\in\mathcal{Y}$, there is 

\begin{equation}\label{eqlip2}
\setlength{\abovedisplayskip}{3pt}
\setlength{\belowdisplayskip}{3pt}
e^{-\epsilon}\le\frac{Pr(G_i=g|Y_i=y,\theta^i_{XG})}{Pr(G_i=g|\theta^i_{XG})}\le{e^{\epsilon}}.
\end{equation}
Denote ${q}_{xy}^i\triangleq\{Pr(Y_i=y|X_i=x)\}$, and $t^i_{gx}\triangleq Pr(X_i=x|G_i=g)$, $\forall x\in\mathcal{X}, g\in\mathcal{G}$, $y\in\mathcal{Y}$. By Bayes rule, the privacy constraints in \eqref{eqlip2} can be expressed as:
\begin{equation}\label{eqlipp2}
\setlength{\abovedisplayskip}{3pt}
\setlength{\belowdisplayskip}{3pt}
e^{-\epsilon}\le\frac{\sum_{x\in\mathcal{X}}q^i_{xy}t^i_{gx}}{\sum_{x\in\mathcal{X}}q^i_{xy}P^i_{x}}\le{e^{\epsilon}}.
\end{equation}
Let $\mathbf{q}^i$ be the set of perturbation probabilities in $\mathcal{M}_i$. Then,
 when $\epsilon$ and each $\mathcal{P}^i_{\mathcal{X}\mathcal{G}}$ are given, the set of inequalities in Eq. (\ref{eqlip2}) forms a feasible region $\mathcal{T}_i$ for $\mathbf{q}^i$,  $\forall{i\in{1,2,..N}}$. 

 The definition of \textit{utility} depends on application scenarios. For example, in statistical aggregation, the estimation accuracy is often measured by absolute error or mean square error \cite{MMSE}\cite{heavyhitter}; in location tracking, it is typically measured by Euclidean distance \cite{Geo}; under information theoretical framework, distortion is typically applied \cite{7498650}. In this paper, we denote $U(S,\hat{S})$ as the utility.
 
In general, there is a tradeoff between utility and privacy. We can   formulate  the following   optimization problem  to find the optimal mechanism that yields the optimal tradeoff:  
\begin{equation}\label{u-p-tradeoff}
\begin{aligned}
&\quad\quad\quad\max U(S,\hat{S}),\\
 &\text{s.t.} ~~~\mathbf{q}^i\in{\mathcal{T}_i},~~ \forall{i\in{1,2,...,N}}. 
 \end{aligned}
\end{equation}
 
\subsection{Problem Formulation}\label{sec:problem}
Focusing on the two applications discussed above, we define utility as the inverse of the Mean Square Error (MSE), which is also adopted in many other works on frequency/histogram estimation \cite{Tianhao,heavyhitter,tianhao_itemset}: 
$U({S},\hat{{S}})=-\mathcal{E}({S},\hat{{S}})$, where $\mathcal{E}({S},\hat{{S}})=E[({S}-\hat{{S}})^2]$. Note that, for weighted summation, the utility is data alphabet dependent while for histogram estimation, it is data alphabet independent, we show how MSE addresses these two different utilities in Sec.~\ref{sec:application}. 
Note that the adversary can use the prior distribution of each user's input data for post-processing. 
From \cite{papoulis2002probability}, it is well-known that the optimal estimator that results in the minimized mean square error (MMSE) is $\hat{{S}}=g(\bar{Y})=E[{S}|\bar{Y}]$. Since $E[E[{S}|\bar{Y}]]=E[{S}]$, $\hat{{S}}$ is an unbiased estimator.  We next formulate the problem under two cases, one is for a fixed prior, the other is for an uncertain prior. 
\subsubsection{Problem formulation for a fixed prior}

Notice that, given each user's prior $\theta^i_{RG}$, the MSE $\mathcal{E}({S},\hat{{S}})$ depends only on each user's perturbation parameters: $\{\mathbf{q}^i\}_{i=1}^N$, as any estimation $\hat{{S}}$ depends on the output $\bar{Y}$ whose distribution is a function of $\{\mathbf{q}^i\}_{i=1}^N$. Thus, maximizing  the utility  is equivalent to finding optimal parameters to minimize the MSE. As a result,  \eqref{u-p-tradeoff} becomes:
\begin{equation}\label{utility}
\begin{aligned}
&\quad\quad\quad\min \mathcal{E}(\mathbf{q}^1,..., \mathbf{q}^N),\\
 &\text{s.t.} ~~~\mathbf{q}^i\in{\mathcal{T}^f_i},~~ \forall{i\in{1,2,..N}},
 \end{aligned}
\end{equation}
\noindent where $\mathcal{T}^f_i$ denotes the feasible region of $\mathbf{q}^i$ for a fixed prior.  

\textbf{Problem Decomposition:} Next, we show the problem defined in Eq. (\ref{utility}) can be decomposed into local optimization problems for each user. 
Since  we assume that each user's input is independent of each other, all the $f(\cdot)$ functions above can be decomposed into local functions $f_i(\cdot)$ of each $R_i$. 
Then, each of them results in an MSE in aggregation, which is denoted by $\mathcal{E}_i=E[(f_i(R_i)-E[f_i(R_i)|Y_i])^2]$ (for the application of histogram, denote $\mathcal{E}^k_i=E[(f^k_i(R_i)-E[f^k_i(R_i)|Y_i])^2]$ as the MSE of aggregating the $k$-th data with $R_i$). The utility defined in \eqref{utility} satisfies decomposition theorem with proof provided in Appendix B of the supplementary document:

\begin{thm}\label{thm:decompose}
The global optimization problem defined in \eqref{utility} can be decomposed into $N$ local optimization problems:
\begin{equation}
\setlength{\abovedisplayskip}{3pt}
\setlength{\belowdisplayskip}{3pt}
\min_{\{(\mathbf{q}^i)\in{\mathcal{T}_i}\}^N_{i=1}}\mathcal{E}(\mathbf{q}^1,..., \mathbf{q}^N)=\sum_{i=1}^N\min_{(\mathbf{q}^i)\in{\mathcal{T}_i}}\mathcal{E}_i(\mathbf{q}^i).
\end{equation}
\end{thm}

\noindent By Theorem~\ref{thm:decompose}, when each local mechanism is optimized, the global MSE of the system achieves its minimum. In addition, each user can perform its local optimization independent of each other, which well suits the local setting.
Now, each local optimization problem incurs an MSE of:
 \begin{equation}\label{largerlaw}
 \begin{aligned}
 \mathcal{E}_i(\mathbf{q}^i)&=E[(f_i(R_i)-E[f_i(R_i)|Y_i])^2]\\
 &=E\{E[(X_i-E[X_i|Y_i])^2|Y_i]\}\\
   &=E[\text{Var}(X_i|Y_i)]\\
     &\overset{(a)}{=}\text{Var}[X_i]-\text{Var}[E(X_i|Y_i)],\\
 \end{aligned}
\end{equation}
where $(a)$ follows the law of total variance.

\textbf{Utility Gain by Observing $\bar{Y}$:} We next compare with the case where no observation of $\bar{Y}$ is available, the goal is to show the utility gain by observing $\bar{Y}$. Since the curator possesses each $\theta^i_X$, to minimize MSE, his optimal local estimator becomes $E[X_i]$. Then, each local MSE becomes:
\begin{equation}
\begin{aligned}
    &\text{Var}[X_i]-\text{Var}[E(X_i)]\\
    =&\text{Var}[X_i]-E[E^2(X_i)]+E^2[E(X_i)]\\
    =&\text{Var}[X_i].
\end{aligned}
\end{equation}
Compared with \eqref{largerlaw}, the utility gain of observing $\bar{Y}$ is due to the term of $\text{Var}[E(X_i|Y_i)]$. Which means when some observations on $Y_i$ are available, the non-negative term of $\text{Var}[E(X_i|Y_i)]$ helps increase data utility.

For the data utility, the MSE of the estimation is a function of the variance of each user's estimator. Define $\hat{X}_i=E[X_i|Y_i]$ as the local estimator for the $i$-th user, and we have $\hat{S}=\sum_{i=1}^N\hat{X}_i$, which follows the user independence assumption. {As each $\text{Var}[X_i]$ is a constant,}
each local optimization problem can be reformulated as:
\begin{equation}\label{probelm}
\begin{aligned}
&\min\mathcal{E}_i(\mathbf{q}^i)\equiv\max \text{Var}(\hat{X}_i),\\
 &~\text{s.t.}~~\text{\eqref{eqlipp2}}. ~~~~~~~~~~\text{s.t.}~\mathbf{q}^i\in{\mathcal{T}^f_i},~~ \forall{i\in{1,2,..N}}.
\end{aligned}
\end{equation}
Which means, the optimal solutions are at the maximum of the variance of the estimator, subject to the LIP constraints. 

\subsubsection{Problem formulation for uncertain prior}

Next, we consider the case where each user has uncertainty on $\theta^i_{RG}$/$\theta^i_{XG}$. Note that under the context-aware setting, it is assumed that the curator/adversary possesses the exact prior distribution. Such scenarios exist when users possess less information about the data and secrets. For example, the curator has recorded a full history of users' previously released data in the server such that the curator can infer each user's prior. Another example is the curator can estimate a global prior for all the users by observing each user's released data. The third example might be, the user is highly correlated with someone (such as family members or close friends) whose data has been collected or compromised. The user's prior then can be inferred by the curator via the correlations.

In the uncertain prior model, the exact prior $\theta^i_{X}$ is not available for each user, so the prior-dependent utility function defined in \eqref{largerlaw} can not be calculated either. In such case, for each user, the local MSE function is determined by his/her perturbation parameters as well as the exact prior distribution, i.e., $\mathcal{E}_i(\mathbf{q}^i)$ in \eqref{largerlaw} becomes $\mathcal{E}_i(\theta^i_X,\mathbf{q}^i)$. 
A feasible minimax strategy for each user is to find the maximized $\mathcal{E}_i(\tilde{\theta}^i_X,\mathbf{q}^i)$ achieved by a prior of $\tilde{\theta}^i_X\in\mathcal{P}^i_{\mathcal{X}\mathcal{G}}$ and find $\mathbf{q}^{i*}$ which minimizes $\mathcal{E}_i(\mathbf{q}^i|\tilde{\theta}^i_X)$. Thus the problem for the $i$-th user becomes:

\begin{equation}\label{eq:uncertain}
\begin{aligned}
&\min_{\mathbf{q}^{i}}\max_{\tilde{\theta}^i_{XG}\in\mathcal{P}^i_{\mathcal{X}\mathcal{G}}} \mathcal{E}_i(\tilde{\theta}^i_{XG},\mathbf{q}^i),\\
 &\text{s.t.} ~~~\mathbf{q}^i\in{\mathcal{T}^u_i}. 
 \end{aligned}
\end{equation}

Note that the feasible region $\mathcal{T}^u_i$ in \eqref{eq:uncertain} is different from $\mathcal{T}^f_i$. It uses BP-LIP's definition, i.e., LIP must be satisfied for a family of priors.
The utility function in Eq.\eqref{eq:uncertain}: $\mathcal{E}_i(\tilde{\theta}^i_X,\mathbf{q}^i)=\text{Var}(X_i)-\text{Var}(\hat{X}^{bp}_i)$, where $\hat{X}^{bp}_i$ is the optimal estimator at the curator. 
As $\text{Var}(X_i)$ depends only on the exact prior of $\theta^i_X$, the goal of each user is still to maximize $\text{Var}(\hat{X}^{bp}_i)$. Thus Eq.\eqref{eq:uncertain} can be further expressed as:
\begin{equation}\label{u_bounded}
\begin{aligned}
   &\max_{\mathbf{q}^i\in{\mathcal{T}_i^u}}\min_{\tilde{\theta}^i_{XG}\in{\mathcal{P}^i_{\mathcal{X}\mathcal{G}}}}\text{Var}[\hat{X}^{bp}_i(\tilde{\theta}^i_{XG},\mathbf{q}^i)].\\
\end{aligned}
\end{equation}





\section{Mechanism Design and Utility-Privacy Tradeoff}\label{tradeoff}

In this Section, we study the utility-privacy tradeoffs under LIP framework. We start with the generalized RR mechanism for the model with a fixed prior. Then, we extend to the model with uncertain prior. After that, we study mechanisms with local hash and unary encoding, followed by a comparison to LDP based mechanisms. Finally, we show how LIP based mechanisms can be applied in real-world applications.

\subsection{Optimal Mechanism for Fixed Prior}
In general, the closed-form optimal solution for the constrained optimization problem of \eqref{probelm} cannot be directly derived. As the number of linear constraints is quadratically proportional to the dimensions of $X_i$ and $G_i$. Also, the valid constraints depend on the concrete prior and correlation.  We numerically present the results and show the properties of the general model in Sec. \ref{sec:sim}.
We next study some useful properties of the problem in \eqref{probelm}. 
For  the privacy constraints, note that:
\begin{equation}
\begin{small}
    \begin{aligned}
        e^{-\epsilon}\le{{\frac{\min_{x\in\mathcal{X}}q^i_{xy}}{\sum_{x\in\mathcal{X}}q^i_{xy}P^i_{x}}}}\le\frac{\sum_{x\in\mathcal{X}}q^i_{xy}t^i_{gx}}{\sum_{x\in\mathcal{X}}q^i_{xy}P^i_{x}}
        \le{\frac{\max_{x\in\mathcal{X}}q^i_{xy}}{\sum_{x\in\mathcal{X}}q^i_{xy}P^i_{x}}}\le{e^{\epsilon}},
    \end{aligned}
    \end{small}
\end{equation}
which means when $Y_i$ is released satisfying $\epsilon$-LIP with respect to $X_i$, the privacy metric in \eqref{eqlipp2} is satisfied automatically. As a result, protecting the privacy of a latent variable rather than the input data enlarges the feasible region of the perturbation parameters, and hence, an increased utility can be achieved.
We next show that, under some conditions, the privacy requirements are met without introducing noise.

\begin{figure*}[t]
\centering 
\subfigure[Mechanism for binary model with latent variable] 
{ \includegraphics[width=3.7cm]{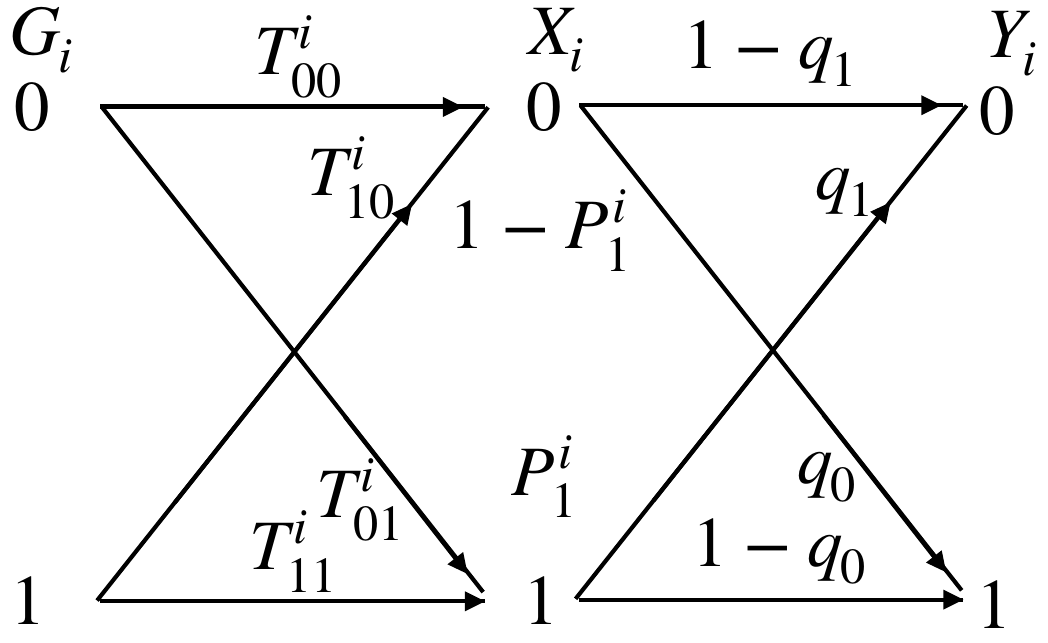} 
\label{fig:General_Model} }
\subfigure[RR Mechanism for M-ary model]
{ \includegraphics[width=3.5cm]{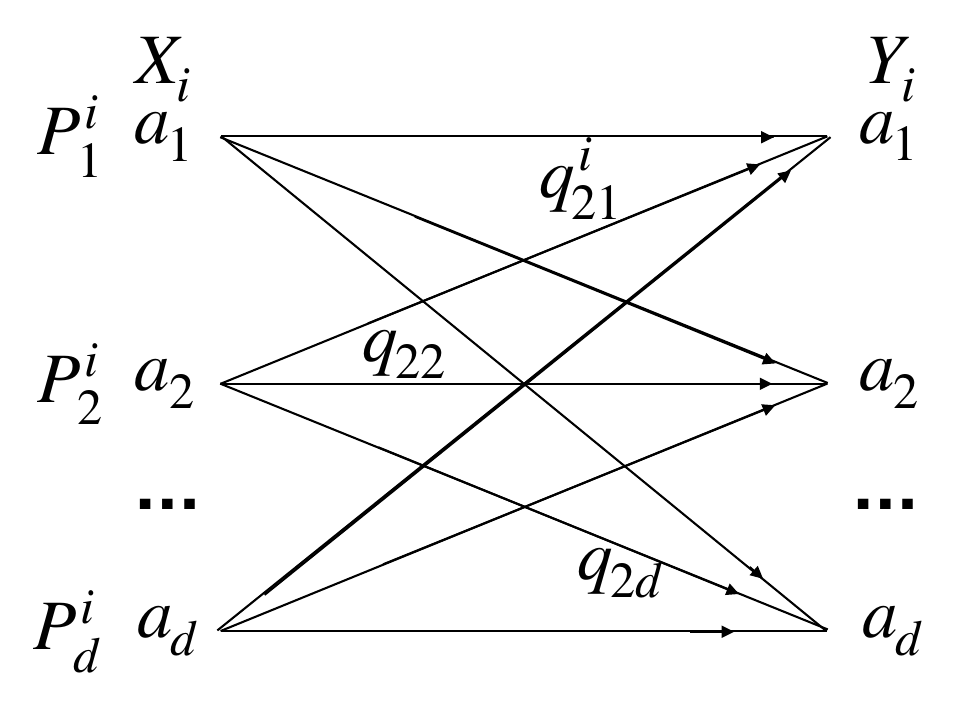} 
\label{fig:MIMO_Model} } 
\subfigure[Mechanism of LH-LIP]
{ \includegraphics[width=4.3cm]{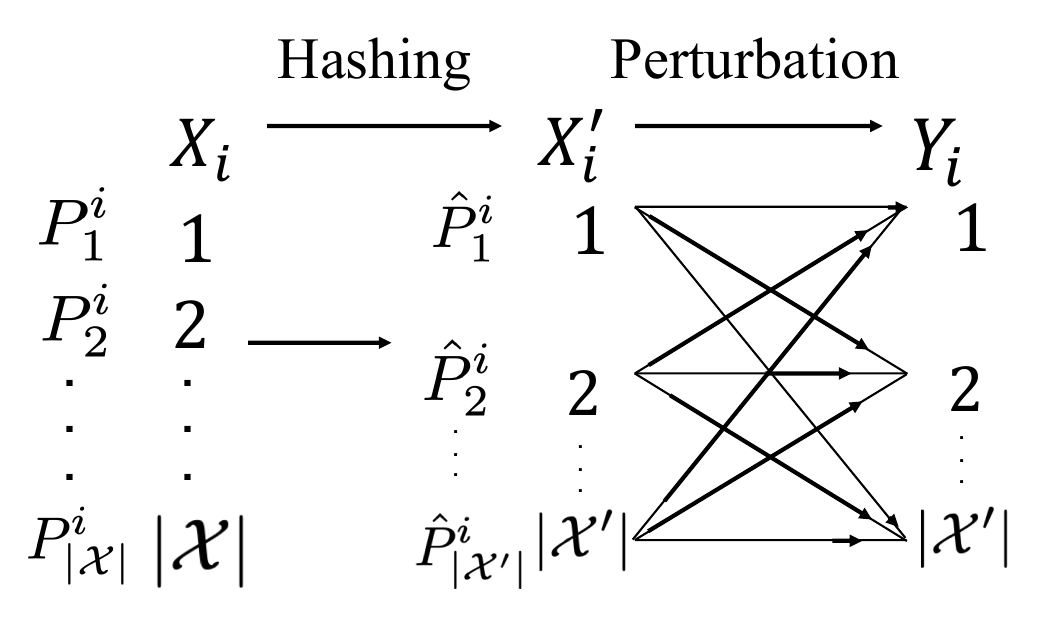}
\label{fig:LH-LIP} }
\subfigure[Mechanism of UE-LIP]
{ \includegraphics[width=5.1cm]{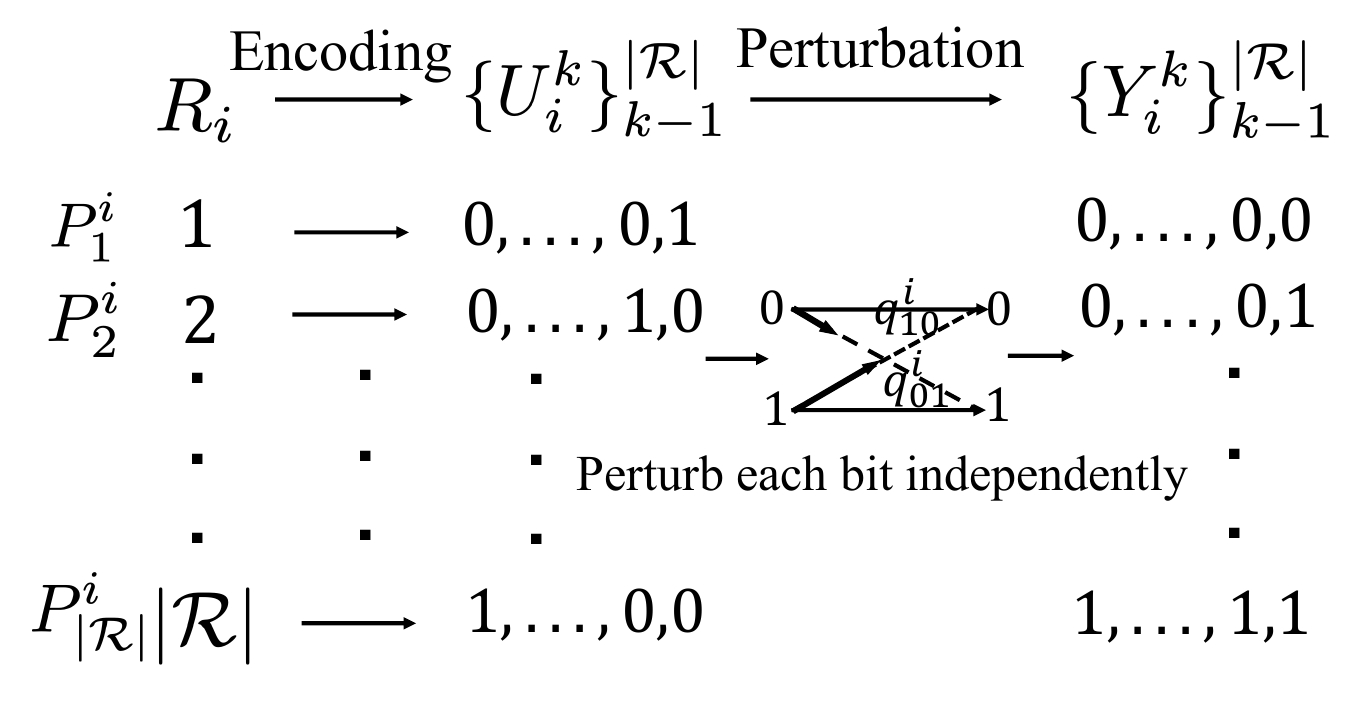}
\label{fig:encoding} }
\caption{Different perturbation mechanisms considered in this paper.} 
\label{Mechanisms} 
\vspace{-17pt}
\end{figure*}

\begin{prop}
For the constrained optimization problem defined in \eqref{probelm}, if for some $a\in{\mathcal{X}}$,
\begin{equation}\label{hiddenbound}
    \max\left\{\frac{\max_{g\in{\mathcal{G}}}t^i_{ga}}{P_a^i}, \frac{P_a^i}{\min_{g\in{\mathcal{G}}}t^i_{ga}}\right\}\le{e^{\epsilon}},
\end{equation}
the optimal $q^{i*}_{ma}=0$ and $q^{i*}_{aa}=1$, $\forall{m\neq{a}}$.
\end{prop}
\begin{proof}
Suppose for some $a\in\mathcal{X}$, $q^i_{aa}=1$ and $q^{i}_{ma}=0$, based on Eq.\eqref{eqlipp2}, $\forall{g\in{\mathcal{G}}}$, there is $e^{-\epsilon}\le\frac{t^i_{ga}}{P_a^i}\le{e^{\epsilon}}$. On the contrary, if this condition is satisfied, to maximize utility, the mechanism decreases $q^i_{ma}$ while increases $q^{i}_{mm}$. In an extreme case, $q^{i*}_{ma}=0$ and $q^i_{aa}=1$.
\end{proof}
Which means if $X_i=a$, the mechanism directly releases $Y_i=X_i$; It is straightforward to extend the result in proposition 1 to: if $\forall{x}\in\mathcal{X}$, $\max\left\{\frac{\max_{g\in{\mathcal{G}}}t^i_{gx}}{P_x^i}, \frac{P^i_x}{\min_{g\in{\mathcal{G}}}t^i_{gx}}\right\}\le{e^{\epsilon}}$, then the mechanism directly releases $Y_i=X_i$.
Notice that, $\forall{x}\in{\mathcal{X}}$ and $\forall{g}\in\mathcal{G}$, the bounded ratio in \eqref{hiddenbound}
equals to $1$ when $X_i$ and $G_i$ are independent, which means directly releasing $X_i$ leaks no information about $G_i$. If the ratio is bounded close to $1$, and the closeness is bounded by [$e^{-\epsilon},e^{\epsilon}$], directly releasing $X_i$ also does not violate LIP. 




\subsubsection{Optimal RR Mechanism under Binary Model}

Next, we derive closed-form optimal solutions for the model with binary input/output. The input is arbitrarily correlated with a binary latent variable. Denote $\mathcal{B}$ as the binary domain of $\{0,1\}$. The binary model is widely used for survey, where each individual's data is first mapped to one bit, then randomly perturbed before publishing to the curator.

{In the binary model, $\mathcal{G}=\mathcal{R}=\mathcal{X}=\mathcal{B}$ (shown in Fig. \ref{fig:General_Model}, we omit $R_i$ for simplicity as $\theta^i_{XG}$ can be calculated given $\theta^i_{R}$, $f_i$ and $T^i_{GR}$)}. $\text{Var}(X_i)$ in \eqref{largerlaw} becomes $P_1^{i}(1-P_1^{i})$.
Denote the perturbation parameters as: $Pr(Y_i=1|X_i=0)=q^{i}_0$, $Pr(Y_i=0|X_i=1)=q^{i}_1$.
Thus, the local MMSE estimator $\hat{X}^b_i$ ($b$ denotes the binary model) becomes:
\begin{equation}
\begin{aligned}
    \hat{X}^b_i=E[X_i|Y_i]
    =&P^{i}_1\left[\frac{q^i_1}{\lambda^i_0}(1-Y_i)+\frac{1-q^i_1}{\lambda^i_1}Y_i\right],
\end{aligned}
\end{equation}
where $\lambda^i_0=Pr(Y_i=0)$ and $\lambda^i_1=Pr(Y_i=1)$.
Then, the utility-privacy tradeoff can be formulated as:
\begin{equation}\label{binary-opt}
\begin{aligned}
            &\max_{(q^i_0,q^i_1)\in{\mathcal{T}^f_i}}{\text{Var}(\hat{X}^b_i)}.\\
 \end{aligned}
\end{equation}
Define $t^{iu}_{g1}=\max_{g\in{\mathcal{G}}}t^{i}_{g1}$; $t^{il}_{g1}=\min_{g\in{\mathcal{G}}}t^{i}_{g1}$. Then the optimal $q^i_1$ and $q^i_0$ correspond to the following Theorem, with proof provided in Appendix C of the supplementary document.

\begin{thm}\label{the}
The optimal $q^i_0$ and $q^i_1$ of the problem defined in \eqref{binary-opt} are:
\begin{equation*}
    \begin{aligned}
    &q^{i*}_0=\max\left\{0,\frac{t^{iu}_{g1}-P^i_1e^{\epsilon}}{(e^{\epsilon}+1)(t^{iu}_{g1}-P^i_1)},\frac{P^i_1-t^{il}_{g1}e^{\epsilon}}{(e^{\epsilon}+1)(P^i_1-t^{il}_{g1})}\right\}\\
    &q_1^{i*}=\max\left\{0,\frac{1+t^{iu}_{g1}e^{\epsilon}-e^{\epsilon}-P^i_1}{(e^{\epsilon}+1)(t^{iu}_{g1}-P^i_1)}, \frac{1+P^i_1e^{\epsilon}-e^{\epsilon}-t^{il}_{g1}}{(e^{\epsilon}+1)(P^i_1-t^{il}_{g1})}\right\}.
    \end{aligned}
\end{equation*}
\end{thm}

Key insight from the binary model with latent variables is, when $X$ is highly correlated with $G$ ($t^{iu}_{g1}$ is large and $t^{il}_{g1}$ is small), $X$ should be privatized with more noise in order to protect $G$; When $X$ is almost independent of $G$ ($t^{iu}_{g1}$ and $t^{il}_{g1}$ are close to $P_1^i$), $X$ can be released with slight perturbation.


\subsubsection{Optimal RR Mechanism under M-ary model when each $R_i=G_i$}\label{sec:tradeoff}
 


Next, we derive the closed-form optimal solutions for the LIP based RR mechanism under M-ary model when each $R_i=G_i$, i.e., the raw data $R_i$ is private. We use M-ary to denote that the input $X_i$ can take multiple possible value. We start from the case where $f_i$ is a bijective or identity function, i.e., $\forall{x\in\mathcal{X}}$ there is only one $r\in\mathcal{R}$, s.t. $f_i(r)=x$. We then extend the optimal solutions to the case where $f_i$ is surjective. {Note that when $f_i$ is bijective, there exists permutation in the mapping from $R_i$ to $X_i$, then the prior of $P^i_x$ equals to the prior of $P^i_r$, where $r=f_i^{-1}(x)$. }Under an RR perturbation mechanism, the perturbation channel and corresponding parameters are shown in Fig. \ref{fig:MIMO_Model}. Denote $\mathcal{X}=\{a_1, a_2,...,a_d\}$,  $Pr(X_i=a_m)=P^i_m$ as the prior distribution of $X_i$,  $Pr(Y_i=a_k)=\lambda^i_k$ as the marginal distribution of $Y_i$. {When $f_i$ is bijective,} the privacy constraints of \eqref{eqlipp2} become, $\forall{m,k\in\{1,2,...,d\}}$:



\begin{equation}\label{simple_cons}
\setlength{\abovedisplayskip}{3pt}
\setlength{\belowdisplayskip}{3pt}
    e^{-\epsilon}\le\frac{q^i_{mk}}{\lambda^i_k}\le{e^{\epsilon}}.
\end{equation}
In the utility function of \eqref{largerlaw}, $\text{Var}[X_i]=\sum^d_{m=1}a_m^2P^i_m-(\sum^d_{m=1}a_mP^i_m)^2$, and the local estimator becomes
\begin{equation}\label{Multiconstraint}
\setlength{\abovedisplayskip}{3pt}
\setlength{\belowdisplayskip}{3pt}
\begin{aligned}
   \hat{X}^m_i=&E[X_i|Y_i]
   =\sum^d_{m=1}a_mPr(X_i=a_m|Y_i)\\
   =&\sum^d_{m=1}\sum^d_{k=1}a_mPr(X_i=a_m|Y_i=a_k)\mathbbm{1}^i_{k},
\end{aligned}
\end{equation}
\noindent where the superscript $m$ denotes the M-ary model, and $\mathbbm{1}^i_{k}$ is the indicator function of $\mathbbm{1}^i_{\{Y_i=a_k\}}$. Then, $\mathbbm{1}^i_{k}$ can be regarded as a binary random variable with the distribution of: $Pr(\mathbbm{1}^i_{k}=1)=\lambda^i_{k}$ and $Pr(\mathbbm{1}^i_{k}=0)=1-\lambda^i_{k}$. As a result: $\text{Var}[\mathbbm{1}^i_k]=\lambda^i_k(1-\lambda^i_k)$ and $\text{Cov}[\mathbbm{1}^i_k,\mathbbm{1}^i_l]=-\lambda^i_k\lambda^i_l$. Taking values in \eqref{Multiconstraint}:
\begin{equation}\label{MIMO_ERROR}
\begin{aligned}
   &\text{Var}(\hat{X}^m_i)
   =\sum^d_{m=1}\sum^d_{n=1}\sum^d_{k=1}a_ma_nq^i_{mk}q^i_{nk}\text{Var}[\mathbbm{1}^i_{k}]\\
 +&\sum^d_{m=1}\sum^d_{n=1}\sum^d_{k=1}\sum^d_{l=1;l\neq{k}}a_ma_nq^i_{mk}q^i_{nl}\text{Cov}[\mathbbm{1}^i_{k},\mathbbm{1}^i_{l}]\\
   =&\sum^d_{m=1}\sum^d_{n=1}a_ma_nP^i_mP^i_n\left(\sum^d_{k=1}\frac{q^i_{mk}q^i_{nk}}{\lambda^i_k}-1\right).
\end{aligned}
\end{equation}
So far, Eq.\eqref{probelm} can be further expressed as $\forall{m,k\in{1,2,...,d}}$:
\begin{equation}\label{multi-opt}
\begin{aligned}
            &\max \text{Var}(\hat{X}^m_i),\\
 &~\text{s.t.} ~ e^{-\epsilon}\le\frac{\lambda^i_k}{q^i_{mk}}\le{e^{\epsilon}}.
 \end{aligned}
\end{equation}

The global optimal solutions follow the next Theorem, with detailed proof provided in Appendix D of the supplementary document.

\begin{thm}[Optimal RR-LIP mechanism under M-ary model]\label{thm3}
For the constrained optimization problem defined in \eqref{multi-opt}, the optimal solutions for the $i$-th user are: $q^{i*}_{mm}=1-(1-P^i_m)/e^{\epsilon}$, $q^{i*}_{mk}=P^i_k/e^{\epsilon}$,  $\forall m,k\in\{1,2,...,d\}$, $m\neq{k}$.
\end{thm}

The constrained optimization problem defined in \eqref{multi-opt} can be visualized in Fig. \ref{fig:contour} (taking a binary example). The curves stand for the contour of $\text{Var}(\hat{X}_i^m)$.  The shaded area stands for the feasible region of $\mathcal{T}^f_i$ for a fixed prior and $\epsilon$. The optimal solutions are found at the boundary of the feasible region, which are intersections of linear equations. 

\begin{figure}[t]
\centering
\includegraphics[width=6cm]{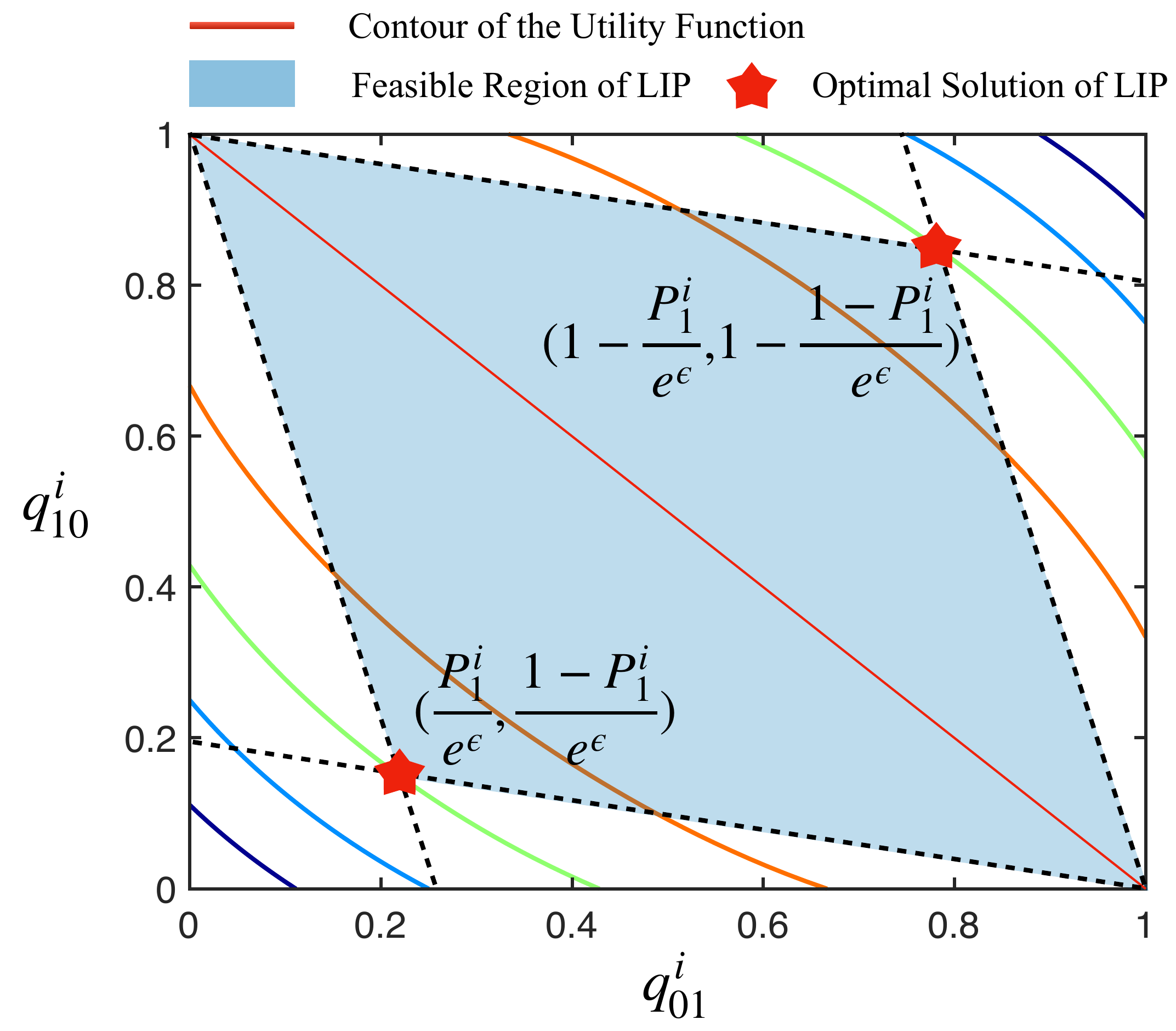}
\setlength{\abovecaptionskip}{-3pt}
\setlength{\belowcaptionskip}{-3pt}
\caption{Illustration of the optimal solutions to the binary LIP model.}
\vspace{-10pt}
\centering
\label{fig:contour}
\end{figure}

From Theorem \ref{thm3}, when $\epsilon$ increases, $\forall{m\in\{1,2,...,d\}}$, all the $q^i_{mm}$s are increasing while all the $q^i_{mk}$s are decreasing ($m\neq{k}$). The value of $q^i_{mk}$s are proportional to $P^i_k$s, i.e., the optimal mechanism is more likely to output the values with larger priors. Note that, the optimal solutions in Theorem 3 are similar to but different from a staircase mechanism\cite{Extreme_ldp} defined for LDP, wherein the likelihood ratio of $\frac{Pr(Y=y|X=x)}{Pr(Y=y|X=x')}$ evaluated at any $x,x'\in\mathcal{X}$, $y\in\mathcal{Y}$ takes value from the set of $\{e^{\epsilon},1,e^{-\epsilon}\}$. The similarity lies in that the maximized utility can be achieved with parameters that just meet privacy constraints. The difference is the solutions in Theorem 3 make most constraints achieve $e^{\epsilon}$, but some of them take values between $[e^{-\epsilon},1]$. We further illustrate the structure of the optimal mechanism through the following example. 

\begin{figure}[t]
\centering
\includegraphics[width=8.9cm]{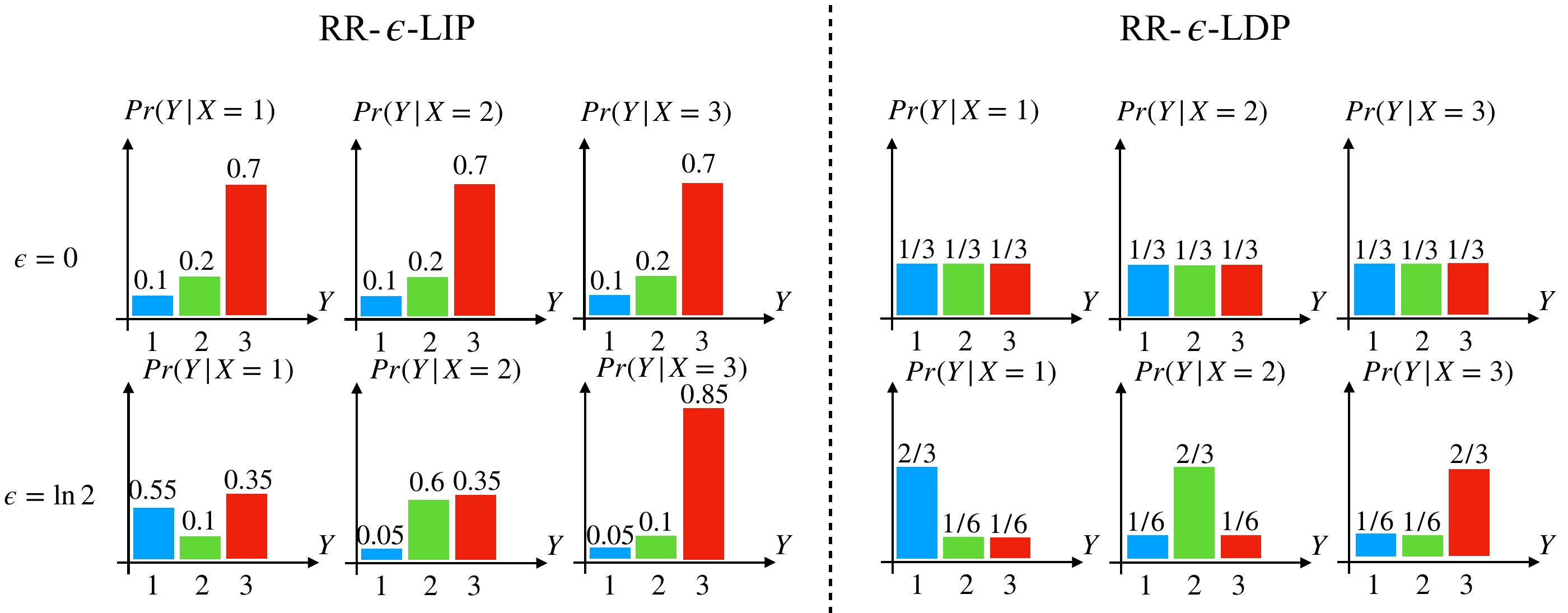}
\setlength{\abovecaptionskip}{-3pt}
\setlength{\belowcaptionskip}{-3pt}
\caption{Illustration of the perturbation parameters under two different $\epsilon$s and a fixed prior.}
\vspace{-13pt}
\centering
\label{fig:perturbation}
\end{figure}

Suppose $\mathcal{X}=\{1,2,3\}$, for the $i$-th user: $P_1=0.1$, $P_2=0.2$, $P_3=0.7$. By Theorem \ref{thm3}, $q^*_{11}=1-0.9/e^{\epsilon}$, $q^*_{22}=1-0.8/e^{\epsilon}$, $q^*_{33}=1-0.3/e^{\epsilon}$, $q^*_{21}=q^*_{31}=0.1/e^{\epsilon}$, $q^*_{12}=q^*_{32}=0.2/e^{\epsilon}$, $q^*_{13}=q^*_{23}=0.7/e^{\epsilon}$. When $\epsilon$ grows, $q^*_{11}$, $q^*_{22}$ and $q^*_{33}$ also increase, which means $X_i$ is more likely to be directly published ($Y_i=X_i$). When $\epsilon$ is small, as ``3" has a larger prior than ``1" and ``2", when $X_i=1$ or $X_i=2$, the mechanism is more likely to output $Y_i=3$ to satisfy the LIP constraints by increasing the posterior of $Pr(X_i=3|Y_i)$. The perturbation parameters are illustrated in Fig. \ref{fig:perturbation}.

{
We next relax the assumption that $f_i$ is bijective and extend to the case where $f_i$ is surjective. The optimal solution is provided in the following Corollary.
{
\begin{cor}
The form of the optimal perturbation parameter of $q^i_{xy}$ when $f_i$ is surjective is identical to that shown in Theorem \ref{thm3}. The difference lies in the prior of $X$: in Theorem \ref{thm3}, $P_m^i =P^i_r$ where $f_i(r)=a_m$; when $f_i$ is surjective, $P^i_m=\sum_{r:f_i(r)=a_m}P^i_r$.  
\end{cor}
\begin{proof}
When $f_i$ is surjective, the privacy metric in \eqref{eqlipp2} can be expressed as:
\begin{equation}\label{eq26}
    \frac{\sum_{x\in\mathcal{X}}Pr(X_i=x|R_i=r)q^i_{xk}}{\lambda^i_k}=\frac{q^i_{mk}}{\lambda^i_k},
\end{equation}
where $f_i(r)=a_m$. For any surjective function $f_i$, $\forall{r\in\mathcal{R}}$, there exists only one $a_m\in\mathcal{X}$ s.t.,  $f_i(r)=a_m$. Therefore, $\forall{r\in\mathcal{R}}$, the ratio of ${\sum_{x\in\mathcal{X}}Pr(X_i=x|R_i=r)q^i_{xk}}/{\lambda^i_k}$ is bounded by $[e^{-\epsilon},e^{\epsilon}]$ is equivalent to $\forall{x\in\mathcal{X}}$, the ratio of $q^i_{mk}/\lambda^i_k$ is bounded by $[e^{-\epsilon},e^{\epsilon}]$. Since the utility definition and the privacy constraints are identical to the optimization problem defined in \eqref{multi-opt}, their optimal solutions are in the same form. 
\end{proof}
}
}

\textbf{Optimal Output Range:} Next, we discuss the optimal output domain of the RR mechanism. Denote $\mathcal{X}=\{a_1,a_2,...,a_d\}$, and $\mathcal{Y}=\{a_1,a_2,...,a_f\}$. The following lemma shows when $d$ is fixed, the optimal $f^*=d$.
\begin{lem}\label{rmk2}
For the RR mechanism under LIP, to minimize the MSE between each $X_i$ and $\hat{X}_i$, when the input range of $d$ is fixed, the optimal output range $f^*$ is $f^*=d$.
\end{lem}
\noindent Detailed proof is shown in Appendix E of the supplementary document. By Lemma \ref{rmk2} we know that enlarge or narrow the output range cannot improve the utility under RR mechanism. 



\subsection{Utility-Privacy Tradeoff for Bounded Priors when each $R_i=G_i$}
The optimal mechanism for each user with bounded priors under M-ary model depends on the concrete $\mathcal{P}^i_{\mathcal{X}\mathcal{G}}$ and therefore, can only be derived numerically. The comparison result is shown in Sec. \ref{sec:sim}.

{Similar to the setting of Sec. \ref{sec:tradeoff}, we next assume that each $R_i=G_i$, which takes value from binary domain $\mathcal{B}$, while the prior comes from a bounded set. It is also straightforward to assume that $f_i$ is bijective for binary model.}
Define the prior uncertainty as $P^i_1=Pr(X_i=1)\in[a,b]$, where $0\le{a}\le{b}\le{1}$. The optimal solutions to the problem defined in Eq.\eqref{u_bounded} correspond to the following proposition:
\begin{prop}\label{BP-LIP}
For the constrained optimization problem defined in \eqref{u_bounded} with binary input/output, the optimal solutions for the $i$-th user are: $q^{i*}_{01}=\frac{b_i}{b_i-a_i+e^{\epsilon}}$ and $q^{i*}_{10}=\frac{1-a_i}{b_i-a_i+e^{\epsilon}}$.
\end{prop}
\begin{proof}
From the proof of Theorem \ref{thm3}, for any $\theta^i_X$, the maximized $\text{Var}(\hat{X}^{bp}_i)$ is achieved at the minimum values of $q^i_{01}$ and $q^i_{10}$, which are found at the boundary of the privacy constraints. They are achieved when $\max_{P^i_1\in[a,b]}\frac{Pr(Y_i=1)}{q^i_{01}}=e^{\epsilon}$ and $\max_{P^i_1\in[a,b]}\frac{Pr(Y_i=0)}{q^i_{10}}=e^{\epsilon}$.
\end{proof}
Observe the expression of $q^{i*}_{01}$ and $q^{i*}_{10}$, when $a_i=b_i=P^i_1$, which means the prior knowledge is certain and fixed, in this case $q^{i*}_{01}=\frac{P^i_1}{e^{\epsilon}}$ and $q^{i*}_{10}=\frac{1-P^i_1}{e^{\epsilon}}$ which are identical to the optimal solutions of Theorem \ref{thm3}; When $a_i=0$, $b_i=1$, we have  the optimal solutions for the WC-LIP: $q^{i*}_{01}=q^{i*}_{10}=\frac{1}{1+e^{\epsilon}}$, which is independent of prior. This result shows that the BP-LIP provides a bridge between the notions of LIP, WC-LIP (LDP) by adjusting prior uncertainty.

\subsection{LIP-based Mechanisms with Encoding}
We next consider other variations of LIP-based mechanisms to mitigate the impact of large input domain on data utility. 
\subsubsection{LIP-based Mechanism with Local Hashing}

 The first method is LIP with Local Hashing (LH-LIP), which can be described as follows:
Denote $\mathcal{H}$ as a universal hash function family such that
each $h\in\mathcal{H}$, maps an input data $X_i$ to $X_i'\in\mathcal{X}'$. Each user randomly selects a hash function $h_i$ from $\mathcal{H}$. 
Then the prior distribution of $X_i'$ can be calculated by combining the input priors according to the hash function:
\begin{equation*}
    Pr(X'=x')=\sum_{x\in\mathcal{X}}Pr(X=x)\mathbbm{1}_{\{h_i(x)=x'\}}
\end{equation*}
Each user then perturbs $X'_i$ by the RR mechanism and outputs $Y_i$, then releases $<Y_i,h_i>$ to the curator. The system model is depicted in Fig. \ref{fig:LH-LIP}.
The curator, after collecting each user's $<Y_i,h_i>$, tries to estimate each local $X_i$.
The privacy metric of LH-LIP when {$R_i=G_i$ and $f_i$ is a bijective function} becomes:
\begin{equation}\label{privacy:LH}
\begin{aligned}
    &\frac{Pr(X_i=x)}{\sum_{x'\in\mathcal{X}'}Pr(X_i=x|X'_i=x')Pr(X'_i=x'|Y_i=y)}\\
    =&\frac{Pr(X_i=x)}{\frac{Pr(X_i=x)}{Pr(X_i'=h_i(x))}Pr(X'_i=h(x)|Y_i=y)}\\
    =&\frac{Pr(X'_i=h_i(x))}{Pr(X'_i=h_i(x)|Y_i=y)},
\end{aligned}
\end{equation}
Notice that, for all $x\in\mathcal{X}$, $y\in\mathcal{Y}$, the ratio in \eqref{privacy:LH} must be bounded by [$e^{-\epsilon},e^{\epsilon}$], which is equivalent to: $\forall{x'\in\mathcal{X}'}$, $y\in\mathcal{Y}$, the ratio in \eqref{privacy:LH} is bounded by [$e^{-\epsilon},e^{\epsilon}$].
The MMSE estimator at the curator becomes (after observing $Y_i=y$):
\begin{small}
\begin{equation}\label{OLH}
\begin{aligned}
&E[X_i|Y_i=y,h_i]\\
=&\sum_{x\in{\mathcal{X}}}x\sum_{x'\in\mathcal{X}'}Pr(X_i=x|X'_i=x',h_i)Pr(X'_i=x'|Y_i=y)\\
=&\sum_{x\in{\mathcal{X}}}x\sum_{x'\in\mathcal{X}'}\mathbbm{1}_{\{h_i(x)=x'\}}\frac{Pr(X_i=x)}{Pr(X'_i=x')}\frac{Pr(X'_i=x')q^i_{x'y}}{\lambda^i_{y}}\\
=&\sum_{x\in{\mathcal{X}}}x\sum_{x'\in\mathcal{X}'}\mathbbm{1}_{\{h_i(x)=x'\}}\frac{Pr(X_i=x)q^i_{x'y}}{\lambda^i_{y}}.\\
\end{aligned}
\end{equation}
\end{small}
Then, the optimization problem for the $i$-the user under LH-LIP  can be formulated as:
{
\begin{equation}\label{problem:LH}
\begin{aligned}
    &\max\text{Var}\left\{E[X_i|Y_i,h_i]\right\}\\
    \text{s.t.}~~ &e^{-\epsilon}\le{\frac{Pr(X'_i=h_i(x))}{Pr(X'_i=h_i(x)|Y_i=y)}}\le{e^{\epsilon}}.
\end{aligned}
\end{equation}}
\begin{prop}[Optimal mechanism for LH-LIP]\label{propLH}
For the constrained optimization problem defined in \eqref{problem:LH}, the optimal solutions are: $q^{i*}_{x'x'}=1-(1-\hat{P}^i_{x'})/e^{\epsilon}$, $q^{i*}_{x'y}=\hat{P}^i_{y}/e^{\epsilon}$, where $\hat{P}^i_{x'}$ denotes the prior distribution of $\sum_{x:h_i(x)=x'}P^i_x$.
\end{prop}
{The results presented in Proposition \ref{propLH} can be directly extended to the case when $f_i$ is surjective by similar derivations to \eqref{eq26}}. Observe that the hashing phase is followed by the RR-LIP, but with a smaller input domain. {It is worth noting that, although the hash function leads to collisions, but due to the fact that each hash function is deterministic, it cannot enhance the privacy measured by LIP.} 

For data utility, collision due to hashing will cause information loss and therefore impact utility. Typically, for small $\epsilon$, the information loss due to the mechanism's perturbation is dominant. Whereas, when $|\mathcal{X}' |$ is small, the collision in hashing dominates the information loss. In \cite{Tianhao}, authors propose Optimal Local Hashing for LDP, which finds the optimal $|\mathcal{X}'|$ under different $\epsilon$s and $|\mathcal{X}|$. However, in context-aware mechanisms, the perturbation parameters depend on the prior. $|\mathcal{X}'|$ only has indirect impact on $q^i_{x'y_i}$ in \eqref{OLH}. As a result, there is no closed-form optimal solution for $|\mathcal{X}' |$. We simulate to study the optimal $|\mathcal{X}'|$ in Section. 6.1.3.

\subsubsection{LIP-based Mechanism with Unary Encoding}\label{sec:uelip}
We next consider another variant of LIP mechanism based on Unary Encoding (UE). From \cite{Tianhao}, we know that, for histogram estimation, incorporating UE in mechanism design could improve the utility-privacy tradeoff of LDP. The intuition behind this improvement is, UE maps a high-dimensional data into a binary vector, the input domain is reduced. On the other hand, in the output vector, multiple locations can be 1, therefore, the input sensitivity is also relaxed. Next, we study LIP based UE mechanism, which can be described as follows.
The UE maps each user's raw data $R_i=r$ into a $|\mathcal{R}|$-bit binary vector with the $r$-th bit equals $1$ and others are zeros. {Note that such local operation on the raw data can be viewed as a local function. To discriminate with other functions, denote $\phi(\cdot)$ as the local function of unary encoding}, and $\{U^k_i\}_{k=1}^{|\mathcal{R}|}=\phi_i(R_i)$ as the encoded vector (input data), with $u_1^{|\mathcal{R}|}\in{\mathcal{B}^{|\mathcal{R}|}}$ as a vector instance. It is worth noting that $|\mathcal{B}^{|\mathcal{R}|}|=|\mathcal{R}|$. Specifically, $U_i^k=u_k\in\{0,1\}$ denotes the $k$-th bit of $U_i$. The mechanism then perturbs each bit independently through a binary RR perturbation channel and releases $\{Y^k_i\}_{k=1}^{|\mathcal{R}|}$. We denote $q^i_{01}$ as the likelihood $Pr(Y_i^k=1|U_i^k=0)$ and $q^i_{10}$ as $Pr(Y_i^k=0|U_i^k=1)$. Note that this may not be optimal, since we assume different bits are perturbed by the same channel. Then the utility function becomes: 
\begin{equation}\label{utility-Unary}
\begin{aligned}
    &E\left[\left(\{U_i^k\}_{k=1}^{|\mathcal{R}|}-E\left[\{U_i^k\}_{k=1}^{|\mathcal{R}|}|\{Y_i^k\}_{k=1}^{|\mathcal{R}|}\right]\right)^2\right]\\
    =&\sum_{k=1}^{|\mathcal{R}|}\left\{\text{Var}[U_i^k]-\text{Var}\left[E[U_i^k|Y_i^k]\right]\right\},
    \end{aligned}
\end{equation}
where $\text{Var}[U_i^k]=P^i_k(1-P^i_k)$ is a constant, and $E[U_i^k|Y_i^k]$ can be expressed as:
\begin{equation}
    \begin{aligned}
    \frac{Pr(R_i=k)(1-q^i_{10})}{Pr(Y_i^k=1)}\mathbbm{1}_{\{Y_i^k=1\}}+\frac{Pr(R_i=k)q^i_{10}}{Pr(Y_i^k=0)}\mathbbm{1}_{\{Y_i^k=0\}}.
    \end{aligned}
\end{equation}
The metric of the privacy constraints can be expressed as:
\begin{equation}\label{eq:unary1}
\begin{aligned}
&{\frac{Pr\left(\{Y^k_i\}_{k=1}^{|\mathcal{R}|}=y_1^{|\mathcal{R}|}\right)}{Pr\left(\{Y^k_i\}_{k=1}^{|\mathcal{R}|}=y_1^{|\mathcal{R}|}|\{U^k_i\}_{k=1}^{|\mathcal{R}|}=u_1^{|\mathcal{R}|}\right)}}.\\
\end{aligned}
\end{equation}
Then, the optimization problem for the $i$-th user under LIP with unary encoding can be formulated as:
\begin{equation}\label{problem_unary}
\begin{aligned}
    &\max_{q^i_{01},q^i_{10}}\sum_{k=1}^{|\mathcal{R}|}\text{Var}[E[U^k_i|Y_i^k]]\\
    &\text{s.t.}~~ e^{-\epsilon}\le{\text{Eq.} ~\eqref{eq:unary1}}\le{e^{\epsilon}}.
\end{aligned}
\end{equation}
The optimal parameters $q^{i*}_{01}$ and $q^{i*}_{10}$ for the above problem are stated in the following Theorem.
The proof is provided in Appendix F of the supplementary document.
\begin{thm}[Optimal mechanism for UE-LIP]\label{thm4}
For the constrained optimization problem defined in \eqref{problem_unary}, the optimal solutions for the $i$-th user are: $q^{i*}_{01}=\frac{1-P^i_{\min}}{e^{\epsilon}-2P^i_{\min}+1}$, and $q^{i*}_{10}=\frac{1}{2}$, where $P^i_{\min}=\min_{r\in\mathcal{R}}P^i_r$.
\end{thm}

\noindent Observe that, each local optimal parameter $q^{i*}_{01}$ depends on the data prior and is a monotonically decreasing function of $P^i_{\min}\in[0,1/|\mathcal{R}|]$. This implies that if the user's prior is uniformly distributed, i.e., $P^i_{\min}=1/|\mathcal{R}|$, the parameter $q^{i*}_{01}$ achieves its minimum, and the utility can be enhanced.

We next extend Theorem \ref{thm4} to consider prior uncertainty. 
When there exist uncertainty on $\theta_X^i$, the local optimization problem for the $i$-th user becomes:
\begin{equation}\label{problem_unary_uncertain}
\begin{aligned}
    &\max_{q^i_{01},q^i_{10}}\min_{\theta^i_R\in\mathcal{P}^i_{\mathcal{R}}}\sum_{k=1}^{|\mathcal{R}|}\text{Var}[E[U^k_i|Y_i^k,\theta_R^i]]\\
    &\text{s.t.}~~ e^{-\epsilon}\le{\text{Eq.} ~\eqref{eq:unary1}}\le{e^{\epsilon}}, \forall{\theta^i_R\in\mathcal{P}^i_{\mathcal{R}}}.
\end{aligned}
\end{equation}

\noindent The optimal $q^{i*}_{01}, q^{i*}_{10}$ are stated in the following corollary:

\begin{cor}
For the constrained optimization problem defined in \eqref{problem_unary_uncertain}, the optimal solutions for the $i$-th user are: $q^{i*}_{01}=\frac{1-\min P^i_{\min}}{e^{\epsilon}-2\min P^i_{\min}+1}$, and $q^{i*}_{10}=\frac{1}{2}$, where $\min P^i_{\min}=\min_{\theta_R^i\in\mathcal{P}^i_{\mathcal{R}}}\min_{r\in\mathcal{R}}P^i_r$.
\end{cor}

\subsection{Comparison with LDP based Mechanism} \label{sec:LDP}
Firstly, we would like to compare the number of privacy constraints in LIP and LDP. The results are summarized in the following remark:
\begin{rmk} \textbf{(Complexity of LDP vs LIP)}.
LDP involves $|\mathcal{Y}||\mathcal{X}| (|\mathcal{X}| - 1)$ linear constraints, while LIP involves $2 |\mathcal{Y}||\mathcal{X}|$ linear constraints. Therefore, when $|\mathcal{X}|>2$, LDP incurs more privacy constraints than LIP.
\end{rmk}

Next, we compare the achievable utilities by the optimal mechanisms based on LDP and LIP.
It is readily seen that the optimal mechanisms proposed in \cite{Tianhao} also apply for the utility functions defined in this paper. The optimal parameters are at the boundary of the privacy constraints. In particular, for RR mechanism, the optimal parameters for LDP are: $\bar{q}^{i*}_{mm}=\frac{e^{\epsilon}}{e^{\epsilon}+|\mathcal{X}|-1}$, $\bar{q}^{i*}_{mk}=\frac{1}{e^{\epsilon}+|\mathcal{X}|-1}$, $\forall{m,k\in{1,2,...,d}}$, $m\neq{k}$. For LDP with Local Hash (LH-LDP), $|\mathcal{X}|$ is changed to $|\mathcal{X'}|$.
For LDP with Optimal Unary Encoding (OUE-LDP), $\bar{q}^{i*}_{10}=1/2$ and $\bar{q}^{i*}_{01}=\frac{1}{e^{\epsilon}+1}$.
 Denote $\mathcal{E}_i^{LIP*}$ as the local MSE from collecting the $i$-th user's data under LIP constraints and $\mathcal{E}_i^{LDP*}$ as that under LDP constraints. Comparing $\mathcal{E}_i^{LIP*}$ with $\mathcal{E}_i^{LDP*}$, we have the following proposition:

\begin{prop}\label{prop:LIPLDP}
Given an arbitrary but fixed prior distribution, $\forall{\epsilon}\in{\mathbf{R}^+}$, there is $\mathcal{E}_i^{LIP*}\le{\mathcal{E}_i^{LDP*}}$.
\end{prop}
\begin{proof}
Since $\mathcal{E}_i^{LIP*}$ and $\mathcal{E}_i^{LDP*}$ are results of the objective function evaluated at different optimal solutions satisfying corresponding privacy constraints. It suffices to show that the optimal perturbation parameters of LDP are within the feasible region of LIP. As $\epsilon$-LDP implies $\epsilon$-LIP, $\forall{\epsilon\ge{0}}$, which means all the $\textbf{q}^i$s that satisfying LDP automatically satisfies LIP.    
\end{proof}
Notice that the curator may take advantage of his prior knowledge to make a further estimation. Nevertheless, LDP based mechanisms suffer a decreased utility than those based on LIP because LIP also utilizes the prior knowledge for mechanism design. Also, note that the optimization problems for LIP and LDP only differ in the feasible regions formed by corresponding privacy constraints. While the feasible region of LDP is fixed for all possible priors, the feasible region of LIP reshapes when the prior changes. 

In particular, we compare the optimal solutions for mechanisms with UE:
\begin{equation}
    q^{i*}_{01}-\bar{q}^{i*}_{01}=\frac{P^i_{\min}(1-e^{\epsilon})}{(e^{\epsilon}-2P_{\min}+1)(e^{\epsilon}+1)}\le{0}.
\end{equation}
The distance diminishes to $0$ if $P_{\min}=0$ (worst-case). Which means UE-LIP will always achieve better utility than OUE-LDP. 
The relationship also applies to BP-LIP and LDP.

\subsection{Real-world Applications of LIP}\label{sec:application}
Next, we discuss how to apply the LIP based mechanisms described above to the following applications. 

\textbf{(Weighted) Summation:}
For weighted summation, 
the aggregated result is $S_{sum}=\sum_{i=1}^N(c_iR_i+b_i)$ with the estimator of $\hat{S}_{sum}=E[S_{sum}|\bar{Y}]$, 
Given any $c_i$ and $b_i$, the MSE becomes:
\begin{equation}\label{eq:sumation}
    \begin{aligned}
        &E[(S_{sum}-\hat{S}_{sum})^2]\\
        =&E\left[\left(\sum_{i=1}^N(c_iR_i+b_i)-E\left[\sum_{i=1}^N(c_iR_i+b_i)|\bar{Y}\right]\right)^2\right]\\
        =&E\left[\left(\sum_{i=1}^N(c_iR_i+b_i)-\sum_{i=1}^NE\left[(c_iR_i+b_i)|{Y_i}\right]\right)^2\right]\\
    \end{aligned}
\end{equation}
Denote $X_i=f_i(R_i)=c_iR_i+b_i$, $\hat{X}^s_i=E[X_i|Y_i]$ ($s$ stands for summation), \eqref{eq:sumation} becomes:
\begin{equation}\label{eq:sumation2}
    \begin{aligned}
        =&E\left[\left(\sum_{i=1}^NX_i-\sum_{i=1}^NE\left[X_i|{Y_i}\right]\right)^2\right]\\
    =&E\left[\sum_{i=1}^N(X_i-\hat{X}^{s}_i)^2+\sum_{i,j=1}^N(X_i-\hat{X}^{s}_i)(X_j-\hat{X}^{s}_j)\right]\\
    \overset{(a)}{=}&\sum_{i=1}^NE[(X_i-\hat{X}^s_i)^2],
    \end{aligned}
\end{equation}
where $(a)$ follows the independent user assumption. 
Note that, when users have uncertain priors, as long as the curator possesses each accurate $\theta_X^i$, he is able to design each local unbiased estimator accordingly, which makes the global utility of $E[(S_{sum}-\hat{S}_{sum})^2]$ decomposable.

So far, the utility function of weighted summation can be  expressed as the form in \eqref{probelm}.
\begin{rmk}
{Each user's local function for (weighted) summation is $f_i(R_i)=c_iR_i+b_i$. For the curator, after observing $Y_i=y$, for RR-LIP mechanism, each optimal local estimator is $\hat{X}^s_i=E[X_i|Y_i=y]=\sum_{x\in\mathcal{X}}q^i_{xy}P^i_x/\lambda^i_y$; for LH-LIP mechanism, $\hat{X}^s_i=E[X_i|Y_i=y,h_i]$ (shown in \eqref{OLH}). }
\end{rmk}
{Note that UE-LIP as a binary encoding based method is inherently designed for frequency estimation (data value-independent), not for value related functions. Therefore, UE-LIP is not appropriate for summation query.}

\textbf{Histogram Estimation}\label{model_applications}
Histogram is useful to estimate or compare the popularity or frequency of some categories. 
We can obtain the estimator of the histogram vector, $\hat{{S}}_{hist}=\{\hat{S}_1,\hat{S}_2,...,\hat{S}_{|\mathcal{R}|}\}=\{E[S_1|\bar{Y}],E[S_2|\bar{Y}],...,E[S_{|\mathcal{R}|}|\bar{Y}]\}$, with each entry $E[S_k|\bar{Y}]$:
\begin{equation}
\setlength{\abovedisplayskip}{3pt}
\setlength{\belowdisplayskip}{3pt}
E\left\{\sum^N_{i=1}\mathbbm{1}_{\{R_i=a_k\}}|\bar{Y}\right\}
    =\sum^N_{i=1}Pr(R_i=a_k|Y_i).
\end{equation}
Thus the mean square error of the estimation is
\begin{equation}\label{Error}
\begin{aligned}
&\sum^{|\mathcal{R}|}_{k=1}E\left[\left(\sum^N_{i=1}\{\mathbbm{1}_{\{R_i=a_k\}}-E[\mathbbm{1}_{\{R_i=a_k\}}|Y_i]\}\right)^2\right]\\
    \overset{(a)}{=}&\sum^{|\mathcal{R}|}_{k=1}
    \sum^N_{i=1}E[(\{\mathbbm{1}_{\{R_i=a_k\}}-E[\mathbbm{1}_{\{R_i=a_k\}}|Y_i]\})^2]\\
    =&\sum^{|\mathcal{R}|}_{k=1}\sum^N_{i=1}\{\text{Var}(\mathbbm{1}_{\{R_i=a_k\}})-\text{Var}(E[\mathbbm{1}_{\{R_i=a_k\}}|Y_i])\}.\\
\end{aligned}
\end{equation}
The (a) of \eqref{Error} is because each user's local error is independent, and the expectation of the unbiased estimator is identical to that of the estimated value. For histogram estimation, $X_i=f_i(R_i)=\{\mathbbm{1}_{\{R_i=a_1\}},\mathbbm{1}_{\{R_i=a_2\}},...,\mathbbm{1}_{\{R_i=a_{|\mathcal{R}|}\}}\}$ (the form is identical to that of $\phi$ for unary encoding studied in Sec. \ref{sec:uelip}), and \eqref{Error} can be expressed as:
\begin{equation*}
    \sum^N_{i=1}\{\text{Var}(X_i)-\text{Var}(E[X_i|Y_i])\}.
\end{equation*}
which is identical to the form in \eqref{probelm}.

\begin{rmk}
{Each user's local function for histogram estimation is: $f_i(R_i)=\phi(R_i)=\{\mathbbm{1}_{\{R_i=a_1\}},\mathbbm{1}_{\{R_i=a_2\}},...,\mathbbm{1}_{\{R_i=a_{|\mathcal{R}|}\}}\}$. With RR-LIP or LH-LIP, given $Y_i=y$, each optimal local estimator at the curator is: $\hat{X}^h_i=\{Pr(R_i=a_1|Y_i=y), Pr(R_i=a_2|Y_i=y),..., Pr(R_i=a_{|\mathcal{R}|}|Y_i=y)\}$; With UE-LIP, given $\{Y^k_i\}_{k=1}^{|\mathcal{R}|}=y_1^{|\mathcal{R}|}$, $\hat{X}^{h}_i=\{Pr(U_i^1=1|Y_i^1=y^1),Pr(U_i^2=1|Y_i^2=y^2),...,Pr(U_i^{|\mathcal{R}|}=1|Y_i^{|\mathcal{R}|}=y^{|\mathcal{R}|})\}$.}
\end{rmk}

\section{Evaluation}\label{sec:sim}
In this Section, we simulate with synthetic and real data to validate our analytical results. In the first part, we validate via Monte-Carlo simulation. We examine the impact on the utility-privacy tradeoff from the prior distribution, data correlation, and input domain. We also consider a model where the utility is measured by Hamming distance instead of MSE. In the second part, we evaluate with real-world datasets:  Gowalla (location check-ins) and Census Income (People income survey).
We evaluate utility by the square root average MSE in order to normalize the influence of user count, also to make it comparable to the absolute error. Note that doing so does not affect the optimalities in any of our optimization problems. In addition, since LIP provides a relaxed privacy guarantee than LDP, it is not easy to compare their utilities under the same privacy guarantee. Thus,  we compare their optimal utilities under any given privacy budget of $\epsilon$. {Since for the experiments with synthetic data, each mechanism takes as input $X_i$ not $R_i$, we directly generate $X_i$ in the following experiments. } 

\subsection{Simulation Results with Synthetic Data}
To generate synthetic data, we consider $5000$ users in the system. We first randomly generate a local prior distribution $\theta^i_X$ for each user and sample each user' input data $X_i$ from $\theta^i_X$. Then, for the model with prior uncertainty, each $\theta^i_X$ is generated for multiple times as the bounded set containing all priors, and the true prior is randomly chosen from this set. Each $X_i$ takes value from domain $\mathcal{X}$ (the default domain for M-ary model is $\mathcal{X}=\{0,1,2,3,4\}$). Each user possesses secret data $G_i$ which also takes value from $\mathcal{X}$ (it can be directly extended to the case where $G_i$ comes from a different domain than $X_i$). Then randomly generate correlation between $X_i$ and $G_i$ (for multiple times as the bounded set). 
\begin{figure}[t]
\centering
\includegraphics[width=6cm]{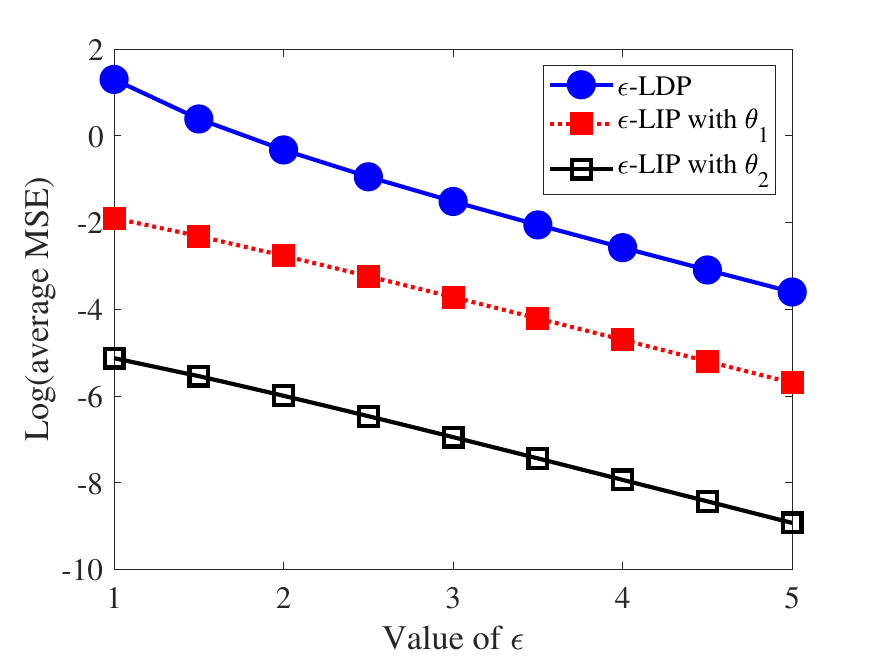}
\setlength{\abovecaptionskip}{-3pt}
\setlength{\belowcaptionskip}{-3pt}
\caption{The utility-privacy tradeoff comparison among prior-aware and prior-free models considering different prior distribution.}
\centering
\label{fig:prior_compare}
\vspace{-10pt}
\end{figure}

\subsubsection{Impact of different prior distributions on utility-privacy tradeoff}\label{sec:sim_1}

Firstly, we would like to demonstrate the impact of different prior distributions on the utility-privacy tradeoff. Let each user share the same prior distribution, and for each user $X_i=G_i$. We consider two sets of priors, one is uniformly distributed $\theta_1=\{0.2, 0.2, 0.2, 0.2, 0.2\}$ and the other is more skewed: ${\theta}_2=\{0.025, 0.025, 0.025, 0.025, 0.9\}$. In addition, we also compare to $\epsilon$-LDP based mechanism with prior-independent estimator $\hat{C}$\cite{Tianhao}. This model treats $X_i$ as instance rather than random variable:
\begin{equation}
\setlength{\abovedisplayskip}{3pt}
\setlength{\belowdisplayskip}{3pt}
\hat{C}=\frac{\sum^N_{i=1}Y_i-N{p_i}}{1-2p_i},
\end{equation}
where $q^i_{mm}=p_i=\frac{e^{\epsilon}}{e^{\epsilon}+|\mathcal{X}|-1}$, is the optimal perturbation parameter, The context-free estimation results in an MSE of:
\begin{equation}
\setlength{\abovedisplayskip}{3pt}
\setlength{\belowdisplayskip}{3pt}
\begin{aligned}
E[(S-\hat{C})^2]=\text{Var}[\hat{C}]
=\frac{N(|\mathcal{X}|-2+e^{\epsilon})}{(e^{\epsilon}-1)^2}.
\end{aligned}
\end{equation}



\noindent The comparison is shown in Fig. \ref{fig:prior_compare}, where $\epsilon$ ranges from $1$ to $5$  with a step of $0.5$.  We can observe that considering the prior in data perturbation and aggregation can largely improve the utility. When each $\theta^i_X=\theta_1$ (prior is uniformly distributed), the utility achieved by $\epsilon$-LIP is decreased than the case when the prior is more skewed, i.e., each $\theta^i_X=\theta_2$. Intuitively, 
 with a skewed prior, users' inputs are highly certain, only considering prior in the estimator can already result in accurate aggregation. As the privacy constraints of LIP with both $\theta_1$ and $\theta_2$ are parameterized by the same $\epsilon$, which means a skewed prior would result in higher utility than a uniformly (or close to uniformly) distributed one under the same privacy guarantee.

 \begin{figure}[t]
\centering
\includegraphics[width=6cm]{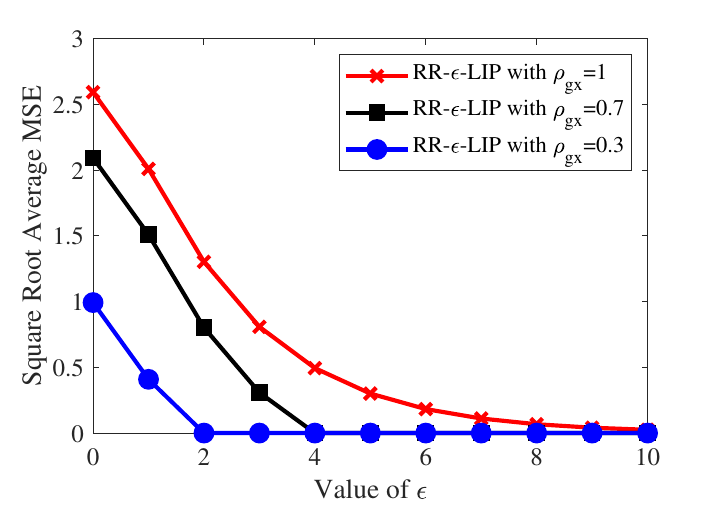}
\setlength{\abovecaptionskip}{-3pt}
\setlength{\belowcaptionskip}{-3pt}
\caption{The impact of correlation between each $X_i$ and $G_i$ to the utility privacy tradeoff provided by LIP.}
\centering
\label{fig:correlation}
\vspace{-10pt}
\end{figure}

 \begin{figure}[t]
\centering
\includegraphics[width=6cm]{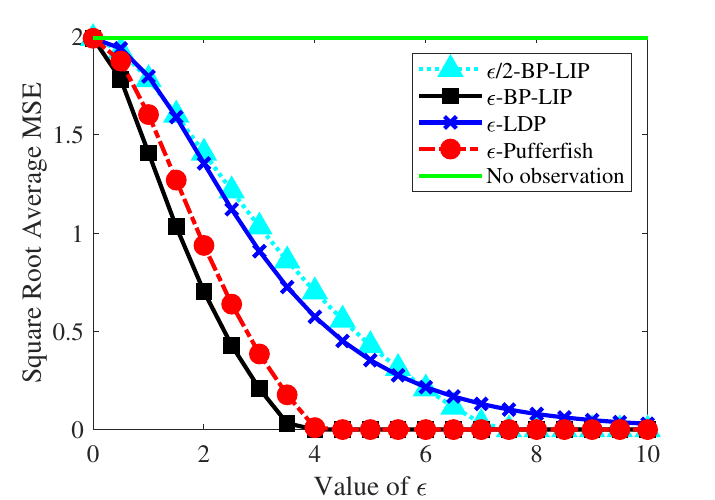}
\setlength{\abovecaptionskip}{-3pt}
\setlength{\belowcaptionskip}{-3pt}
\caption{Utility-Privacy tradeoff comparison with bounded prior among different privacy notions.}
\centering
\label{fig:U_P_latent}
\vspace{-10pt}
\end{figure}

\begin{figure}[t]
\centering 
\subfigure[Utility comparison when $|\mathcal{X}|$ increases from $10$ to $50$, $\epsilon$ is fixed to be $\epsilon=1$] 
{ \includegraphics[width=6cm]{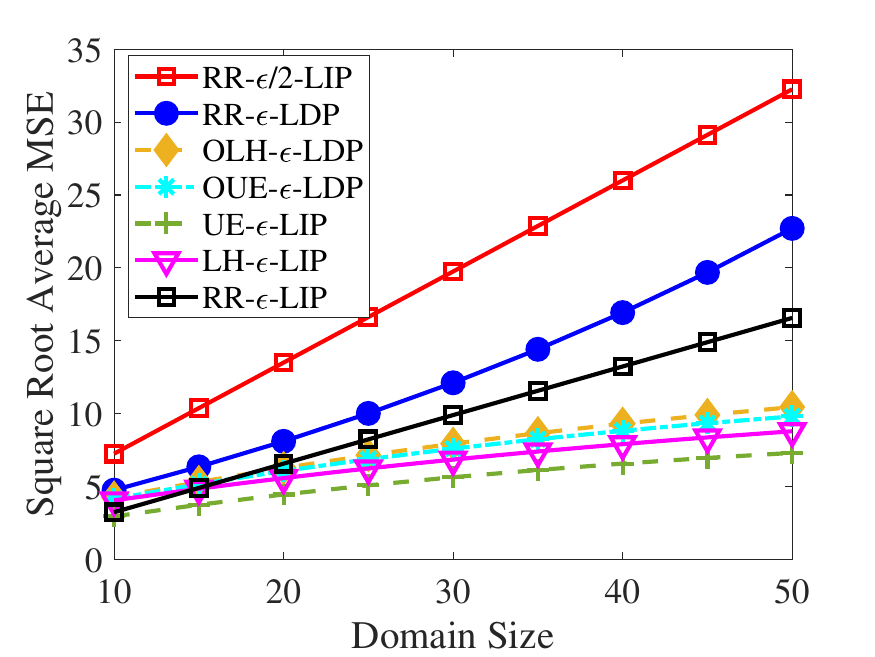} 
\label{fig_domain} }\\
\subfigure[Utility comparison for LH-$\epsilon$-LIP with different hashing sizes given a skewed prior, $\epsilon=1$, $|\mathcal{X}|=20$.] 
{ \includegraphics[width=6cm]{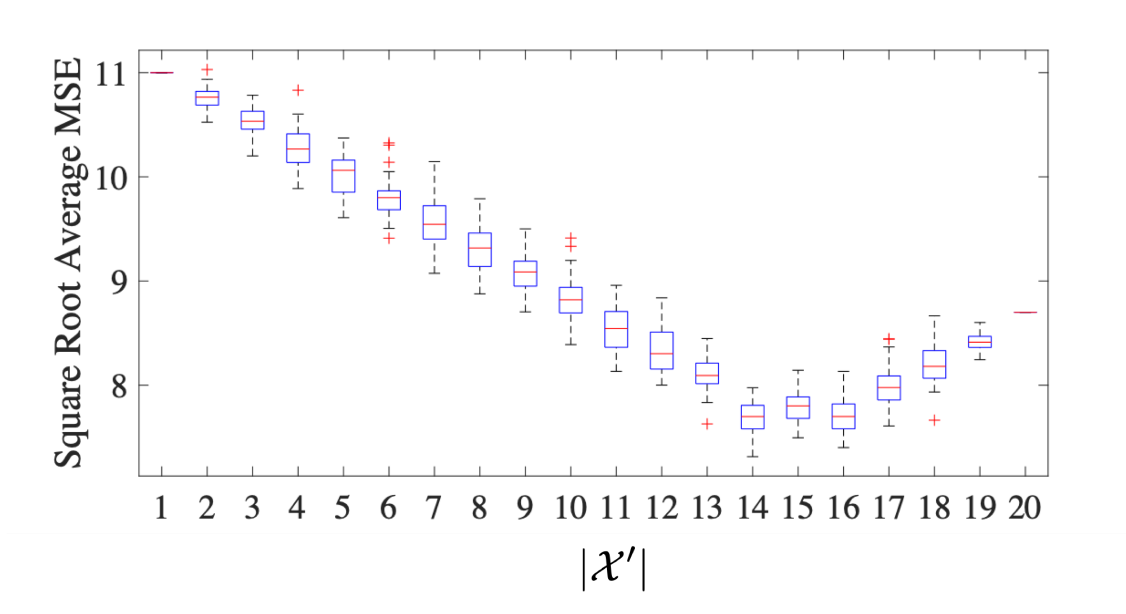} 
\label{fig:domain2} } 
\caption{Impact of domain size on data utility and optimal hashing domain.} 
\label{fig:single_model_compare} 
\vspace{-10pt}
\end{figure}



\begin{figure}[t]
\centering 
\subfigure[Comparison between LIP and LDP when utility is measured by Hamming distance with binary data] 
{ \includegraphics[width=6cm]{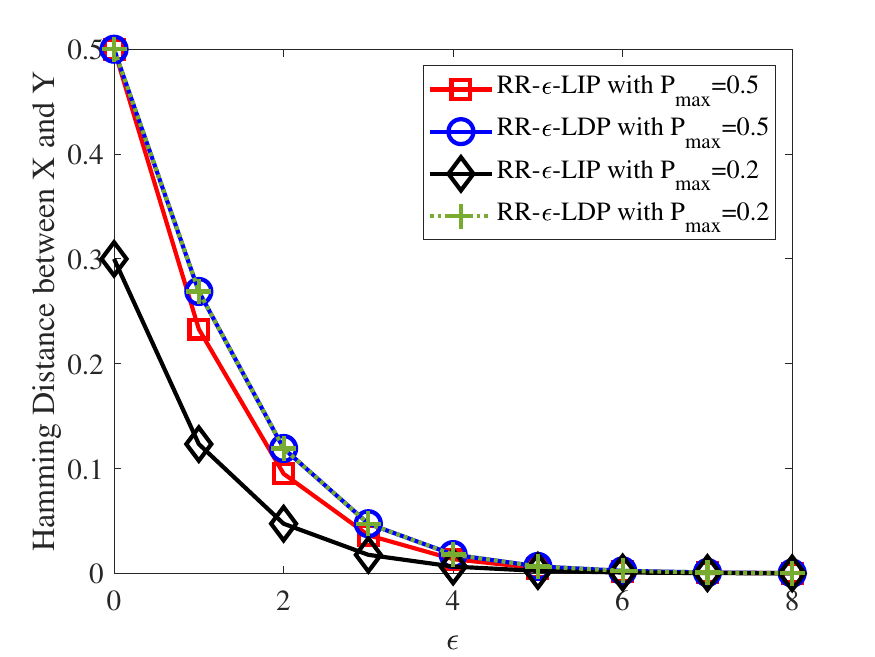} 
\label{hamming1} }\\
\subfigure[Comparison between LIP and LDP when utility is measured by Hamming distance with $|\mathcal{X}|=5$] 
{ \includegraphics[width=6cm]{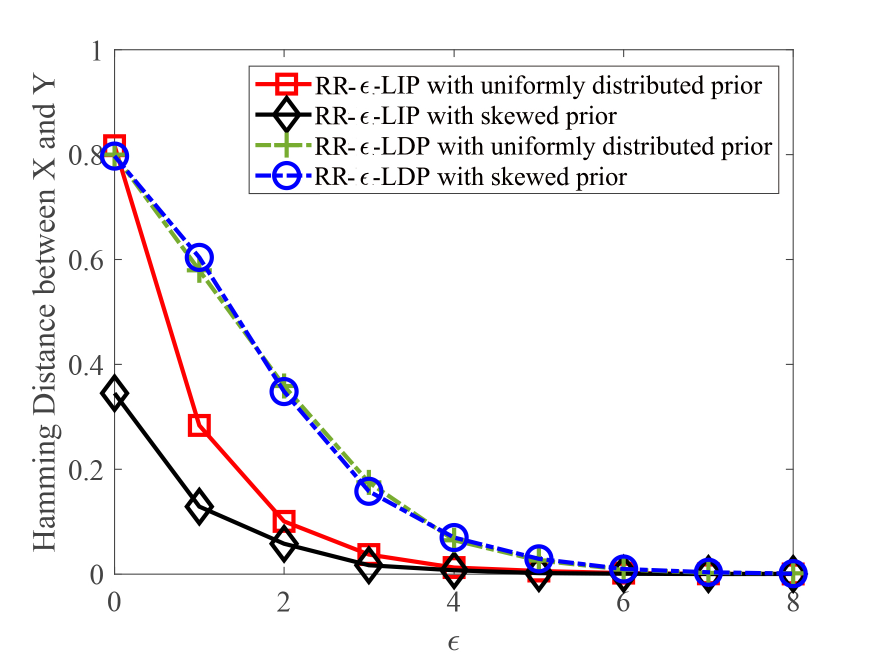} 
\label{hamming2} } 
\caption{Utility-privacy tradeoff comparison from a rate-distortion perspective} 
\label{fig:HD} 
\vspace{-10pt}
\end{figure}

\begin{figure*}[t]
\centering 
\subfigure[Utility-privacy tradeoffs for location histogram estimation (users are i.i.d. and domain size is equivalent to $|\mathcal{X}|=83$).]
{ \includegraphics[width=6cm]{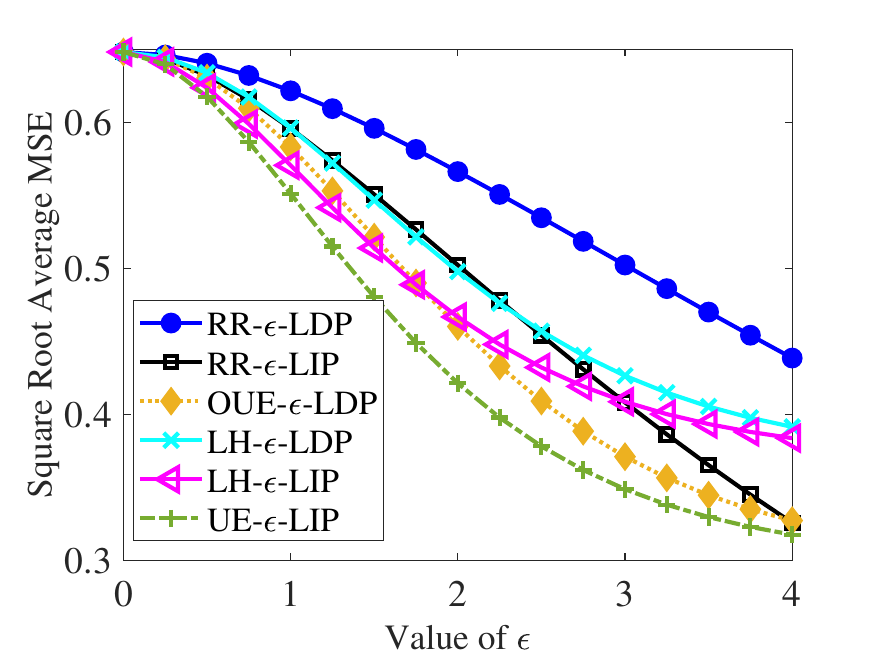} 
\label{fig:Gowalla} }\qquad\qquad
\subfigure[Utility-privacy tradeoffs for work class aggregation while protecting annual income privacy (Model with hidden variable, $|\mathcal{R}|=4$ and $|\mathcal{X}|=4$).] 
{ \includegraphics[width=6cm]{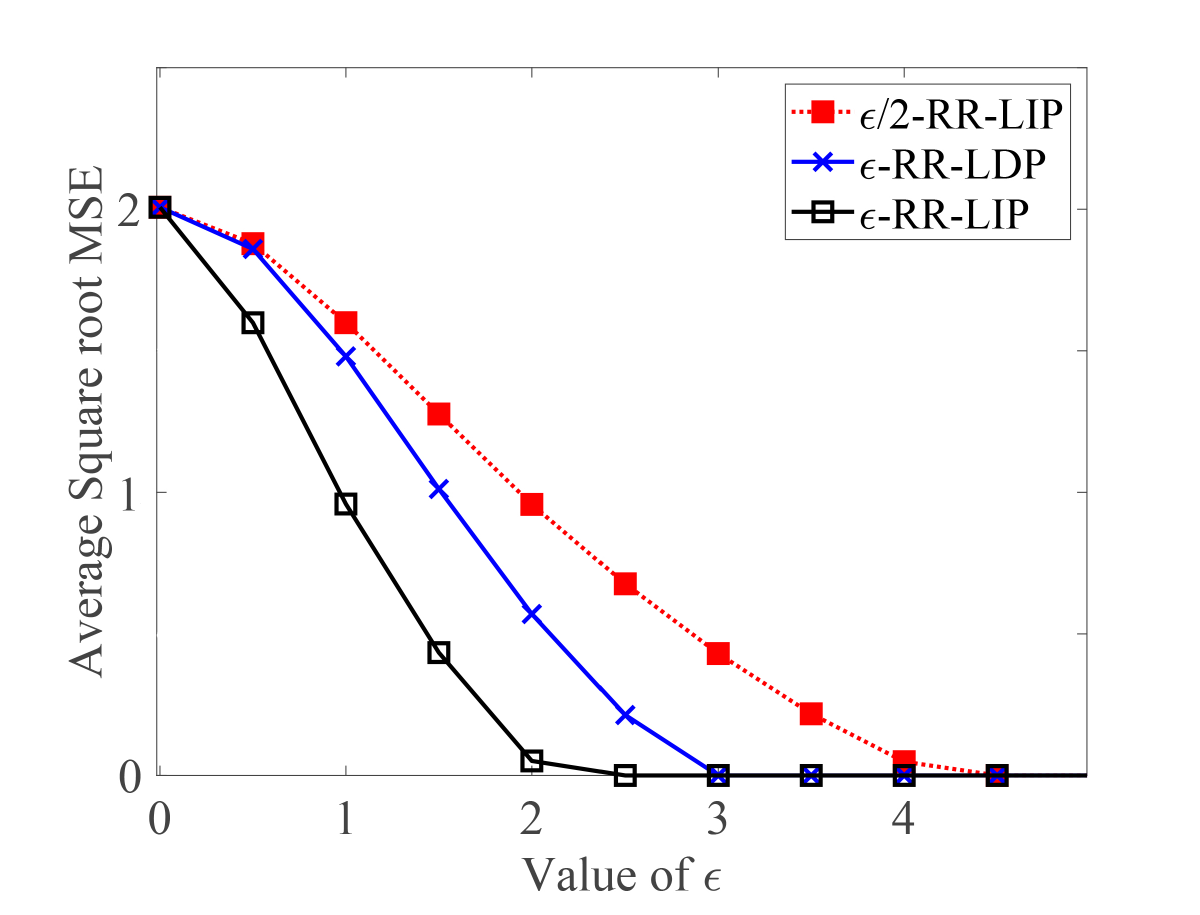} 
\label{fig:adult} } 
\caption{Utility-privacy tradeoff comparisons using real-world data.} 
\label{fig:single_model_compare} 
\vspace{-17pt}
\end{figure*}

\subsubsection{Utility as a function of correlation with latent variable}

We next consider the model with each $G_i\neq{X_i}$. We first examine the utility as a function of the correlation between $X_i$ and $G_i$, and consider a fixed prior of $X_i$. The correlation between $X_i$ and $G_i$ is measured by the correlation coefficient $\rho_{gx}=\frac{\sigma_{xg}}{\sigma_{x}\sigma_g}$, we then find the conditional probability of $T^i_{GX}$ by fixing $\rho_{gx}$ to be $1$, $0.7$ and $0.3$ respectively (when $\rho_{gx}>0$, a larger $\rho_{gx}$ implies stronger correlation between $X_i$ and $G_i$). Under each correlation, we derive the utility-privacy tradeoff provided by RR-$\epsilon$-LIP based mechanism. The result is shown in Fig. \ref{fig:correlation}.
Observer that, stronger correlation results in decreased utility compared to weaker correlation. The reason lies in that when the correlation is strong, more noise is needed to privatize the input data $X_i$. When $\rho_{gx}=1$, $X_i=G_i$, the mechanism cannot achieve zero MSE. When $\rho_{gx}=0.3$, given any $\epsilon\ge{2}$, the MSE is decreased to $0$. Because no noise is added to perturb $X_i$ and the correlation between $X_i$ and $G_i$ makes $G_i$ hard enough to be inferred.

\subsubsection{Comparison among different privacy notions with latent variable and uncertain prior}\label{Sec:6.1.3}

Next, we consider the scenario where each user's input data $X_i$ is correlated to $G_i$ with correlation comes from bounded set $\mathcal{P}_i$. We then compare the utility provided by the following privacy notions under the RR mechanism: (a). $\epsilon$-LIP ($\epsilon/2$-LIP) with bounded prior;  (b). $\epsilon$-Pufferfish privacy; (c). $\epsilon$-LDP;  (d) with no observations on $\bar{Y}$. Note that $\epsilon$-LDP provides privacy protection against the worst-case prior, including $X_i=G_i$. From the impact of correlation between $X_i$ and $G_i$ on data utility, we know that, for LDP based mechanisms, protecting $G_i$ and $X_i$ are equivalent. 

The utility-privacy tradeoff comparisons are shown in Fig. \ref{fig:U_P_latent}. Observe that different mechanisms share the same start point because the prior distribution of each user's input data $X_i$ is fixed and known to the curator. Even though BP-LIP and Pufferfish privacy have larger feasible regions for perturbation parameters by considering the bounded set of correlations between $X_i$ and $G_i$, as long as the input data is not independent of the latent variable, the mechanism needs to make $X_i$ and $Y_i$ independent in order to achieve zero privacy leakage. The utility provided by LIP increases faster with $\epsilon$ than Pufferfish and LDP, because the feasible regions of pufferfish and LDP are within that of LIP. Another observation is that the utility of $\epsilon$-LDP is not bounded between $\epsilon/2$ and $\epsilon$-LIP, as we have shown in Section 3. This is because LIP further considers the correlation between $X_i$ and $G_i$, while LDP considers the worst-case correlation, which could be $1$, i.e., $X_i=G_i$. Finally, we can observe the utility gain by using outputs from the privacy-preserving mechanism compared to the case when only using prior for estimation. Observe that for different $\epsilon$s, taking no observations results in a constant MES which equals the variance of the data.

\subsubsection{Impact of Domain Size on Models}\label{sec:monte}

Next, we compare how the data domain impacts the utility-privacy tradeoff of LIP and LDP: Consider each $X_i$ in the system has a domain size from $|\mathcal{X}|=10$ to $|\mathcal{X}|=50$. We then fix $\epsilon=1$ and show the utilities with different input domain sizes. The goal is to compare the utility provided by RR mechanisms and encoding based mechanisms. To this end, we also compare with other variations of LDP based mechanisms, which improve RR-LDP's performance significantly when $|\mathcal{X}|>3e^{\epsilon}+2$ (\cite{Tianhao}). One is LDP with Optimal Unary Encoding, and the other is LDP with Optimal Local Hashing. From \cite{Tianhao}, the optimal hashing size is $|\mathcal{X}'|^*=e^{\epsilon}+1$. When $\epsilon=1$, we have $|\mathcal{X}'|^*=4$. To make a fair comparison, we consider each $X_i=G_i$, and we compare with LH-$\epsilon$-LIP when $|\mathcal{X}'|=4$. The utility comparison as a factor of the input domain is shown in Fig.\ref{fig_domain}.

From Fig.\ref{fig_domain}, we have the following insights: (1) When the correlation between $X_i$ and $G_i$ is not considered, the utility provided by RR-$\epsilon$-LDP is always sandwiched between RR-$\epsilon$-LIP and RR-$\epsilon/2$-LIP under any domain size. Because they share the same utility function, and the utility depends on the size of the parameters' feasible regions. (2) When $|\mathcal{X}|$ is small ($|\mathcal{X}|<15$), RR-LIP provides better utility than LH-LIP. (3) 
For large $|\mathcal{X}|$, UE-LIP outperforms RR-LIP, and the gap enlarges as $|\mathcal{X}|$ increases. LDP based mechanisms have similar trends. (4) UE-LIP always provides better utility than OUE-LDP, LH-LIP always outperforms OLH-LDP. The reasons are described in Section~\ref{sec:LDP}.

Further, we compare the utility provided by LH-LIP with different hashing sizes. We fix $\epsilon=1$ and consider $|\mathcal{X}|=20$ with a prior of $[0, 0.005,0.001,...,0.095,0.1]$ (the increment is $0.005$). We then range $|\mathcal{X}'|$ from $1$ to $20$. Given different hashing sizes, there could be multiple hash functions. When there exist more than $100$ hash functions, we randomly select $100$ functions and calculate their corresponding utilities. In Fig. \ref{fig:domain2}, we show the utility comparison among LH-LIP with different hashing sizes. Observe that, under each $|\mathcal{X}'|$, utilities varies for different hash functions, because different hash functions imply different prior combinations. Intuitively, when $\mathcal{X}'$ is uniformly distributed, more noise is added in perturbation than when the distribution of $\mathcal{X}'$ is skewed. Also, observe that the optimal hashing size should be around $12$ to $16$. However, when $|\mathcal{X}'|\in[12,16]$, there still exist some hash functions that provide poor utilities. Such observation further confirms that the optimal hashing size cannot be determined under an arbitrary prior.

\subsubsection{Comparison between LIP and LDP for Hamming distance-based utility}
Next, we compare $\epsilon$-LIP to $\epsilon$-LDP when the utility is measured by Hamming distance between each input $X_i$ and output $Y_i$, i.e.,
\begin{equation}
    \text{Utility}=-\sum_{i=1}^N||Y_i-X_i||_h,
\end{equation}
where $||A-B||_h=0$ if $A=B$, $||A-B||_h=1$ if $A\neq{B}$. Hamming distance is usually adopted in a rate-distortion framework, where rate measures the privacy leakage and distortion captures data utility. In \cite{Extreme_ldp}, an optimal mechanism is derived under LDP constraints.

We next compare LIP and LDP under two cases: (1) Binary model with uncertain prior: when each input data $X_i$ is binary and is sampled from $\theta^i_X$. Notice that $\theta^i_X$ can be further specified by $P^i_1$. It is assumed that the exact $P^i_1$ is unknown to each user, but each of them knows that $P^i_1$ is upper bounded by $P_{\max}=\max{P}^i_1$. Then each user's released data $Y_i$ is generated by a RR mechanism satisfying $\epsilon$-BP-LIP described in Section 4.2 or RR-$\epsilon$-LDP. (2) When each $X_i$ takes value from $\mathcal{X}$ ($|\mathcal{X}|=5$) with a fixed prior. The prior is assumed to be known by each user. We consider two scenarios on data prior: when data is uniformly distributed or data has a skewed prior. Then each user's released data $Y_i$ is generated by RR-$\epsilon$-LIP or RR-$\epsilon$-LDP.
The utility comparison is shown in Fig. \ref{fig:HD}. Observe that RR-$\epsilon$-LIP provides better utility than RR-$\epsilon$-LDP under each case, and when the prior is more skewed,  the advantage becomes even enhanced.

 \subsection{Simulation with Real-world Datasets}

\subsubsection{Histogram Estimation with Location Check-In Dataset}
In this subsection, we compare the performance of different models with the real-world dataset Gowalla, a social networking application where users share their locations by checking-in. There are 6,442,892 users in this dataset. For each user, a trace of the check-in locations is recorded. {Denote the $i$-th user's location trace as $(R_i^1,R_i^2,...,R_i^k)$, where the superscript denotes different check-ins, and for different users, $k$ can be different.} 
With this dataset, we intend to estimate a histogram of users' latest check-in location. {It is assumed that the past location trace of $(R_i^1,R_i^2,...,R_i^{k-1})$ has already been released, and both the users and the curator can use $\{R_i^1,R_i^2,...,R_i^{k-1}\}_{i=1}^N$ to calculate a global prior of the latest check-in location.} We first divide the area into $36\times36$ districts, then map each user's {latest check-in location, which is denoted as $R_i^k=X_i$} into districts. As we studied in Section~\ref{model_applications}, for each user, the latest check-in location is perturbed according to the LIP (LDP) based mechanisms, and a random vector estimator is used for the curator to estimate the histogram. 

The results are shown in Fig. \ref{fig:Gowalla}. Observe that the utilities provided by different mechanisms increase more slowly than the results in Section~\ref{Sec:6.1.3}. This is because each $X_i$ has a larger domain size in this experiment, and the prior of each district is very small. Hence, increasing $\epsilon$ has less influence on the utility than when each data value has a larger prior. Also, note that RR-LIP provides decreased utility than LH-LIP and UE-LIP when $\epsilon$ is small. But eventually, when $\epsilon$ increases, RR-LIP outperforms UE-LIP and LH-LIP, because $q^{i*}_{10}$ in UE-LIP is fixed to be $1/2$. When $\epsilon$ increases, all $0$s in the vector tend to be directly released, but the $1$ in the vector still has a one-half probability of being perturbed as $0$. Also, in LH-LIP, when $\epsilon$ increases, the information loss at hashing affects the utility more than at perturbation. Finally, UE-LIP provides better utility than OUE-LDP, but the gap diminishes as $\epsilon$ increases. 


\subsubsection{Latent Variable Privacy with Dataset of Annual Income}

Next, we testify our analysis of the model with latent variables by simulation on a real-world dataset: ``Census income" (Adult dataset), a census survey dataset in which 48842 users' personal information is listed, including 14 attributes, such as age, work class, marriage, race, gender, education, and annual income, {which are denoted as $\{R_i^1,R_i^2,...,R_i^{14}\}$ respectively}. We assume each user's data is published and collected independently. In the field of machine learning, the Adult dataset is usually used for predicting whether each user's annual income is over 50k dollars by training on all the personal information (taken as features). In this experiment, we want to aggregate users' work classes while protecting annual incomes. {In this dataset, the raw data $R_i^2$, work class, has a domain size of 8}: \{Private, Self-emp-not-inc, Self-emp-inc, Federal-gov, Local-gov, State-gov, Without-pay, Never-worked\}. Each user's annual income, {$G_i=R_i^{14}$}, also has a domain size of 8: \{below 20K, 20k-30k, 30k-40k, 40k-50k, 50k-60k, 60k-70k, 70k-80k, over 80k\}, We use number $0$ to $7$ to stand for each of them and statistically calculate the frequency of each value to be the priors. We then find the correlation between each user's work class and the annual income by deep learning (a built-in network of Tensorflow). {In this experiment, we consider the input data $X_i$ has a smaller domain size than $|\mathcal{R}_i|$, i.e., $f_i$ is surjective but not bijective: let $\mathcal{X}$ be \{Private, Self-employed, Government, Never-worked\}. The prior of $X_i$ and correlation with $G_i$ can be calculated by the mapping rule. }Then each user publishes his/her $X_i$ by the LIP/ LDP based mechanism with perturbation parameters numerically solved by the optimization problem defined in \eqref{probelm}. The comparison is shown in Fig. \ref{fig:adult}.
From Fig. \ref{fig:adult}, we observe that the proposed $\epsilon$-LIP model provides better utility than $\epsilon$-LDP. Compared with Monte-Carlo simulations, with this dataset, each model requires a larger $\epsilon$ to diminish to $0$, because the latent variable $G$ is highly correlated with $X$.

From the experimental results, we have the following insights: a)  context-aware privacy notions provide better utility than context-free notions, and when the prior is more skewed, the advantage becomes even enhanced; b) LIP based mechanism achieves better utility than those based on LDP when using the same prior dependent estimator, the utility gain lies in measuring the prior knowledge in the privacy notion. c) When the data domain increases, the utility under each notion decreases. Incorporating encoding in the mechanism improves utility when $\epsilon$ is small. d) Utilities of the models with latent variables are higher than those without because the collected data becomes less sensitive. When the correlation between $ X $ and $ G $ is weak, for some $ \epsilon $, $ X $ can be directly published to achieve zero MSE.
\section{Conclusion}\label{sec:con}
In this paper, the notion of local information privacy is proposed and studied. As a context-aware privacy notion, it provides a relaxed privacy guarantee than LDP by introducing prior knowledge in the privacy definition while achieving increased utility. We implement the proposed LIP notion into the data aggregation framework and derive the utility-privacy tradeoff, which minimizes the MSE between the input data and the estimation while protecting the privacy of the raw data or a private latent variable that is correlated with the input data. We consider different scenarios on the prior availability (uncertainty) and data correlation. We also incorporate the encoding methods into the mechanism to mitigate the influence of a large input data domain. Finally, we use synthetic and real-world data to demonstrate the impact of data prior, correlation, and data domain, and compare the utility provided by proposed mechanisms to those based on LDP. Results show that LIP based mechanisms provide better utility than those based on LDP.

\ifCLASSOPTIONcaptionsoff
  \newpage
\fi



%
\bibliographystyle{ieeetr}
\bibliography{Ref}

\appendices
\section{Proof of Lemma 3}\label{sec:app1}
\begin{proof}
When $\epsilon$-LIP is satisfied, the privacy metric of DI can be expressed as:
\begin{equation*}
\begin{aligned}
    &\frac{Pr(Y=y|X=x)Pr(X=x)}{Pr(Y=y|X=x')Pr(X=x')}\\
    &\le{\frac{Pr(Y=y)Pr(X=x)e^{\epsilon}}{Pr(Y=y)Pr(X=x')e^{-\epsilon}}}\\
    &\le e^{2\epsilon+D^{{X}}_{\infty}}.
\end{aligned}
\end{equation*}
For the other direction, when $\epsilon$-DI holds, we have:
\begin{equation*}
    \begin{aligned}
    &\frac{Pr(Y=y|X=x)}{Pr(Y=y|X=x')}\le{e^{\epsilon+D^{{X}}_{\infty}}}.\\
    \end{aligned}
\end{equation*}
Then we have:
\begin{equation*}
    \begin{aligned}
    Pr(Y=y)=&\sum_{x\in{\mathcal{X}}}Pr(Y=y|X=x)Pr(X=x)\\
    \le&{\sum_{x\in{\mathcal{X}}}e^{\epsilon+D^{{X}}_{\infty}}Pr(Y=y|X=x')Pr(X=x)}\\
    \le&{e^{\epsilon+D_{\infty}}Pr(Y=y|X=x')}.
    \end{aligned}
\end{equation*}
Similarly, $ Pr(Y=y)\ge{e^{-\epsilon-D^{{X}}_{\infty}}Pr(Y=y|X=x')}$. Thus $(\epsilon+D^{{X}}_{\infty})$-LIP is satisfied.
\end{proof}

\section{Proof of Theorem 1}\label{app_thm1}
\begin{proof}
The MMSE estimator $\hat{\mathbf{S}}$ can be expressed as:

\begin{equation}\label{estimator}
\setlength{\abovedisplayskip}{3pt}
\setlength{\belowdisplayskip}{3pt}
\begin{aligned}
&E[{S}|\bar{Y}]=E[f(\bar{R})|\bar{Y}]
=E[f(R_1,R_2,...,R_N)|\bar{Y}]\\
\overset{(a)}{=}&E[f_1(R_1)|\bar{Y}]+E[f_2(R_2)|\bar{Y}],...,+E[f_N(R_N)|\bar{Y}]\}\\
\overset{(b)}{=}&\sum^{N}_{i=1}E[f_i(R_i)|Y_i],
\end{aligned}
\end{equation}

where (a) in Eq. \eqref{estimator} is due to the independence of $R_i$s, and (b) is because $R_i$ is only correlated with $Y_i$ in the output sequence.  Thus, $\mathcal{E}({{S},\hat{{S}}})$  can be derived as:
\begin{small}
\begin{equation}\label{error_1}
\mathcal{E}({{S},\hat{{S}}})=E\left[\left(\sum^{N}_{i=1}\{f_i(R_i)-E[f_i(R_i)|Y_i]\}\right)^2\right].
\end{equation}
\end{small}

Note that, for the application of histogram, the error forms an error vector of $(S_k,\hat{S_k})_{k=1}^d$. By the definition of second order norm. The mean square error of this case is: 
\begin{small}
\begin{equation*}
\mathcal{E}(S_k,\hat{S_k})_{k=1}^d=\sum^d_{k=1}E\left[\left(\sum^{N}_{i=1}\{f^k_i(R_i)-E[f^k_i(R_i)|Y_i]\}\right)^2\right],
\end{equation*}
\end{small}
where $f^k_i(R_i)=\mathbbm{1}_{\{R_i=k\}}$.

We next show that in general, the total MSE can be decomposed into the summation of local MSEs.
\begin{equation*}
\begin{aligned}
&\mathcal{E}({{S},\hat{{S}}})=E\left[\left(\sum^{N}_{i=1}\{f_i(R_i)-E[f_i(R_i)|Y_i]\}\right)^2\right]\\
&=\sum^{N}_{i=1}E\left[f_i(R_i)-E[f_i(R_i)|Y_i]\right]^2\\
&-2\sum^{N}_{j=1,l\neq{j}}E\{(f_j(R_j)-E[f_j(R_j)|Y_j])(f_l(R_l)-E[f_l(R_l)|Y_l])\}.\\
\end{aligned}
\end{equation*}
The cross terms are 0 because $\forall{j,l}\in\{1,...,N\}$ and $j\neq{l}$:
\begin{equation*}
\begin{aligned}
&E\{(f_j(R_j)-E[f_j(R_j)|Y_j])(f_l(R_l)-E[f_l(R_l)|Y_l])\}]\\
=&E[(f_j(R_j)-E[f_j(R_j)|Y_j])]E[(f_l(R_l)-E[f_l(R_l)|Y_l])]\\
=&[E(f_j(R_j))-E\{E[f_j(R_j)|Y_j]\}][E(f_l(R_l))-E\{E[f_l(R_l)|Y_l]\}],
\end{aligned}
\end{equation*}
where $E(f_j(R_j))-E\{E[f_j(R_j)|Y_j]\}$ and $E(f_l(R_l))-E\{E[f_l(R_l)|Y_l]\}$ are 0, because the estimator is unbiased. Thus, $\mathcal{E}({{S},\hat{{S}}})=\sum^N_{i=1}\mathcal{E}_i(\mathbf{q}^i)$. 

We next show that the global optimal solutions (perturbation parameters) satisfy each local privacy constraint: 

Assume that for each user, the minimized $\mathcal{E}_i(\mathbf{q}^i)=e_i$ is achieved at $\mathbf{q}^{i*}\in{\mathcal{T}_i}$, then $\mathcal{E}(\mathbf{q}^{1*},...,\mathbf{q}^{N*})=\sum_{i=1}^Ne_i$. If for some user ``k" who takes parameters $\mathbf{q}^{k}\in{\mathcal{T}_k}$, by assumption, we know that $\mathcal{E}_k(\mathbf{q}^k)\ge{e_k}$. Thus,
\begin{equation*}
\setlength{\abovedisplayskip}{3pt}
\setlength{\belowdisplayskip}{3pt}
\sum_{i=1}^k\mathcal{E}_i(\mathbf{q}^{i*})+\mathcal{E}_k(\mathbf{q}^k)+\sum_{i=k+1}^N\mathcal{E}_i(\mathbf{q}^{i*})\ge{\sum_{i=1}^Ne_i}.
\end{equation*}
That means the minimal value of $\mathcal{E}(\mathbf{q}^1,...,\mathbf{q}^N)$, where $\mathbf{q}^i\in{\mathcal{T}_i}$, 
$\forall{i\in{[1,N]}}$ can be achieved if for each user, $\mathbf{q}^i=\mathbf{q}^{i*}$.

\end{proof}

\section{Proof of Theorem 2}
\begin{proof}
The first step is to show the minimal MSE is achieved when $q_0$ and $q_1$ are at their minimum, which can be proved by taking derivative of the MSE function with respect to $q$s to show that MSE is increasing with $q$s.

The second step is to find the minimum values of $q$s, which are found according to the privacy constraints.
To derive the monotocity of the privacy metric with respect to $q$s.
Define $F^i_1=\frac{Pr(G_i=g|Y_i=1)}{Pr(G_i=g)}$, $F^i_2=\frac{Pr(G_i=g|Y_i=0)}{Pr(G_i=g)}$ which can be further expressed as
\begin{equation}
\begin{aligned}
    F^i_1=&\frac{Pr(Y_i=0|G_i=g)}{Pr(Y_i=0)}
    =\frac{(1-q^i_0)T^i_{g0}+q^i_1t^{i}_{g1}}{q^i_1P^i_1+(1-q^i_0)(1-P^i_1)};\\
    F^i_2=&\frac{Pr(Y_i=1|G_i=g)}{Pr(Y_i=1)}
    =\frac{q^i_0T^i_{g0}+(1-q^i_1)t^{i}_{g1}}{(1-q^i_1)P^i_1+q^i_0(1-P^i_1)}.
\end{aligned}
\end{equation}

Taking derivative over $q^i_0$ and $q^i_1$, we have:$\frac{\partial{F^i_1}}{\partial{q^i_0}}=\frac{(t^{i}_{g1}-P^i_1)q^i_1}{(q^i_1P^i_1+(1-q^i_0)(1-P^i_1))^2}$,
    $\frac{\partial{F^i_1}}{\partial{q^i_1}}=\frac{(t^{i}_{g1}-P^i_1)(1-q^i_0)}{(q^i_1P^i_1+(1-q^i_0)(1-P^i_1))^2}$,
    $\frac{\partial{F^i_2}}{\partial{q^i_0}}=\frac{(P^i_1-t^{i}_{g1})(1-q^i_1)}{(1-q^i_1)P^i_1+q^i_0(1-P^i_1)^2}$,
    $\frac{\partial{F^i_2}}{\partial{q^i_1}}=\frac{(P^i_1-t^{i}_{g1})q^i_0}{(1-q^i_1)P^i_1+q^i_0(1-P^i_1)^2}.$

So we know, when $t^{i}_{g1}>{P^i_1}$, $F^i_1$ is monotonically increasing with $q^i$, whereas $F^i_2$ is monotonically decreasing with $q^i$, so the minimum $q^i$s are achieved when $F^i_1=e^{-\epsilon}$ and $F^i_2=e^{\epsilon}$. Solving the equations, and we get:
$q^i_0=\frac{t^{i}_{g1}-P^i_1e^{\epsilon}}{(e^{\epsilon}+1)(t^{i}_{g1}-P^i_1)}$; $q^i_1=\frac{1+t^{i}_{g1}e^{\epsilon}-e^{\epsilon}-P^i_1}{(e^{\epsilon}+1)(t^{i}_{g1}-P^i_1)}$;
When $t^{i}_{g1}<{P_1}$, $F^i_1$ is monotonically decreasing with $q^i$, whereas $F^i_2$ is monotonically increasing with $q^i$, so the minimum $q^i$s are achieved when $F^i_1=e^{\epsilon}$ and $F^i_2=e^{-\epsilon}$. Solving the equation, and we get:
$q^i_0=\frac{P^i_1-t^{i}_{g1}e^{\epsilon}}{(e^{\epsilon}+1)(P^i_1-t^{i}_{g1})}$; $q^i_1= \frac{1+P^i_1e^{\epsilon}-e^{\epsilon}-t^{i}_{g1}}{(e^{\epsilon}+1)(P^i_1-t^{i}_{g1})}$.

The final step is to test the value of $q^i_0$ and $q^i_1$ as functions of $t^{i}_{g1}$. Taking derivative on $q^i$s, we have that the first set of solutions are monotonically increasing with $t^{i}_{g1}$, and the second set of solutions are monotonically decreasing with $t^{i}_{g1}$. Thus, to find a pair of $q_0$ and $q_1$ satisfying $T_{g1}$ for all $g\in{\mathcal{G}}$, we take the maximum of all possible values. As $q$s are non-negative, another candidate in the max function is 0.


\end{proof}

\section{Proof of Theorem 3}\label{OPT_SOL}
\begin{proof}
Notice that $\text{Var}[X_i]$ is a non-negative constant, thus minimizing MSE is equivalent to maximize $\text{Var}[\hat{X_i}]$.

Step 1. Regardless of the privacy constraints:

\textbf{Minimized solution:}\\
Consider a set of parameters: $\mathbf{q_{min}^i}$, when $q^i_{nk}=\lambda^i_k$, $\forall{n,k\in{1,2,3...d}}$, $\text{Var}[\hat{X_i}]=0$. Since $\text{Var}[\hat{X_i}]\ge0$, thus the solution of $q^i_{nk}=\lambda^i_k$ results in a minimal value of  $\text{Var}[\hat{X_i}]$.

\textbf{Maximized solution:} Consider a set of parameters: $\mathbf{q_{max}^i}$, assume that for all $k={1,2...d}$, $q^i_{kk}=1$ and $q^i_{kl}=0$ for all $l\neq{k}$. Under this solution, $\lambda^i_k=P^i_k$ and

\begin{equation}
\begin{aligned}
    &\sum^d_{m=1}\sum^d_{n=1}\sum^d_{k=1}a_ma_nP^i_mP^i_nq^i_{mk}\left(\frac{q^i_{nk}}{\lambda^i_k}-1\right)\\
    =&\sum^d_{n=1}a^2_nP^i_n(1-P^i_n)-\sum^d_{n=1}\sum^d_{m\neq{n}}a_na_mP^i_nP^i_m
    =\text{Var}[X_i].
    \end{aligned}
\end{equation}

Notice that $\mathcal{E}_i\ge{0}$, $\text{Var}[X_i]\ge{\text{Var}[\hat{X_i}]}$. Thus, the solution of $q^i_{kk}=1$ and $q^i_{kl}=0$, $\forall k={1,2,...,d},l\neq{k}$ results in the maximum value of $\text{Var}(\hat{X}_i)$.

Next, investigate the monotonicity of the region between minimum and maximum:

Taking derivative with respect to $q^i_{lk}$, $\frac{\partial{\text{Var}[\hat{X_i}]}}{\partial{q^i_{lk}}}$ becomes
\begin{small}
\begin{equation}\label{der}
\begin{aligned}
   &\frac{1}{(\lambda^i_k)^2}\left[a_l\lambda^i_k\left(2\sum^d_{m=1}(a_mq^i_{mk}-a_j\lambda^i_k)\right)-P^i_l\left(\sum^d_{m=1}a_mq^i_{mk}\right)^2\right]\\
   =&\frac{a_{l}q^i_{lk}\left(\sum_{m\neq{l}}^da_mq^i_{mk}\right)(1-P^i_k)(q^i_{lk}-\lambda^i_k)}{\lambda^i_k}.
\end{aligned}
\end{equation}
\end{small}
From Eq. \eqref{der}, we can observe that the station point of $q^i_{lk}$ is $\lambda^i_k$, which we know is the minimal value and $\text{Var}[\hat{X_i}]$ is monotonically increasing when $q^i_{lk}>\lambda^i_k$;  $\text{Var}[\hat{X_i}]$ is monotonically decreasing when $q^i_{lk}<\lambda^i_k$. As a result, without considering the privacy constraints, the optimal solutions of each $q^i_{mn}$ is either $0$ or $1$. We next show that the maximum value of $\text{Var}[\hat{X_i}]$ can only be achieved by the solutions discussed above.

Now, assume that for the data value $l$, there is a subset of index $\mathcal{S}$ s.t: $q^i_{lk}\neq{1}\neq{0}$, for any $k\in\mathcal{S}$. Denote $\hat{X}$ as the estimator using $\mathbf{q_{max}^i}$ and $\hat{X}'$ as the estimator using $\mathbf{q_{max}^i}$ but the parameters for data value $l$ are substituted according to the  subset. Regardless of the constraints, compare with the variance of $\text{Var}[\hat{X}_i]$ and $\text{Var}[\hat{X}'_i]$, we have:
\begin{equation}
\begin{aligned}
    &\text{Var}[\hat{X}_i]-\text{Var}[\hat{X'_i}]\\=&\sum^n_{k=1}a^2_lP^i_l(\frac{P^i_l}{P^i_l+P^i_k})+\sum^n_{k=1}a^2_kP^i_k(\frac{P^i_k}{P^i_l+P^i_k})\\
   +& \sum^d_{m\notin\{1,2,...,n\}}a_lP^i_la_mP^i_m-2\sum^n_{k=1}a_la_k\frac{P^i_lP^i_k}{P^i_l+P^i_k}\\
    =&\sum^n_{k=1}\frac{(a_lP^i_l-a_kP^i_k)^2}{P^i_l+P^i_k}+\sum^d_{m\notin\{1,2,...,n\}}a_lP^i_la_mP^i_m>{0}.
\end{aligned}
\end{equation}

Thus, the form of the optimal solution is unique: for any $k\in\{1,2,...,d\}$, only one of the $q^i_{kj}=1$, other $q^i_{kj}=0$. 

Step 2. With privacy constraints:

As $\text{Var}[\hat{X_i}]$  is monotonically increasing when $q^i_{lk}>\lambda^i_k$; and monotonically decreasing when $q^i_{lk}<\lambda^i_k$.
The optimal solution (with privacy constraints) lies on the boundaries of the constraints:     $e^{-\epsilon}=\frac{\lambda^i_k}{q^i_{jk}}$, or $\frac{\lambda^i_k}{q^i_{jk}}={e^{\epsilon}}$ (under $0\le{q^i_{jk}}$; $\sum^d_{n=1}q^i_{jn}=1;$, $\forall{j,k\in{1,2,...,d}}$).

When one of the probabilities of $q^i_{m1}, q^i_{m2},..., q^i_{md}$, approaches 1 and others approaches 0, there are $d$ possible selections, and consider all the $m\in\{1,2,...,d\}$ there are $d!$ feasible solutions. We now consider the case where $q^i_{kk}$s approach 1 for all $k\in{1,2,...,d}$, and other $q^i_{kj}$s are approaching  0. For the $q^i_{kk}$s which approach 1, the upper bounds is valid, and for $q^i_{kj}$s which approach 0, the lower bounds are valid. Considering the privacy constraints, we know the upper bound of $q^i_{kk}$ is $\lambda^i_k/e^{-\epsilon}$ and the lower bound of $q^i_{kj}$ is $\lambda^i_k/e^{\epsilon}$. As  $q^i_{kk}+\sum_{j=1,j\neq{k}}^dq^i_{kj}=1$, for all $j$s $q^i_{kj}$s are approaching boundaries simultaneously, as a result, they may not reach the boundaries at the same time.

Next, discuss whether lower bounds or upper bounds are reached first.
When lower bounds are reached, $q^i_{jk}=\frac{\lambda_k}{e^{\epsilon}}$ for all $ j\neq{k}$. Thus $q^i_{kk}=1-(1-P^i_k)/e^{\epsilon}$, $\lambda^i_k=P^i_k$.

We can check whether $q^i_{kk}$s are in the feasible region:

\begin{equation}
\begin{aligned}
    \frac{\lambda^i_k}{q^i_{kk}}-e^{-\epsilon}=&\frac{e^{\epsilon}P^i_k}{e^{\epsilon}+P^i_k-1}-e^{-\epsilon}\ge{0},
    \end{aligned}
\end{equation}

\begin{equation}
\begin{aligned}
    e^{\epsilon}-\frac{\lambda^i_k}{q^i_{kk}}=&e^{\epsilon}-\frac{e^{\epsilon}P^i_k}{e^{\epsilon}+P^i_k-1}\ge{0}.
    \end{aligned}
\end{equation}

So, when $q^i_{kj}$s reach the lower bound, $q^i_{kk}$ is still in the feasible region. It is readily seen that when $q^i_{kk}$ reaches the upper bound, $q^i_{kj}$s do not satisfy the privacy constraints. 
\end{proof}

\section{Proof of Lemma 4}\label{OPT-range}
\begin{proof}
As the MSE is the difference between the variance of the input data and the variance of the estimator, when $d$ is fixed, the variance of the input data is fixed. It is equivalent to show when $f\neq{d}$, the variance of the estimator decreases.

 We know the optimal solution of the parameters of any input $X_i=a_k$ are in the form of $q^i_{kk}$ is approaching $1$ while other $q^i_{kj}$s are approaching $0$ so that each input value can be inferred by a particular output. For example, given $Y_i=a_k$, one can probably infer that $X_i$ is also $a_k$ and the confidence increases with $\epsilon$. 

\textbf{when $f<d$}, when the $d$ is fixed, $\text{Var}(X)$ is also fixed. denote  $\text{Var}(\hat{X_i})$ as the variance of the estimator with $d=f$ and $\text{Var}(\hat{X'_i})$ as the variance of the estimator with $d>f$.
     Recall that 
     \begin{small}
     \begin{equation}
        \hat{X_i}=\sum^d_{j=1}\sum^d_{k=1}a_jPr(X_i=a_j|Y_i=a_k)\mathbbm{1}^i_{k},  
     \end{equation}
     \begin{equation}
        \hat{X'_i}=\sum^d_{j=1}\sum^f_{k=1}a_jPr(X_i=a_j|Y_i=a_k)\mathbbm{1}^i_{k},  
     \end{equation}
     \end{small}
     First assume that for each $j\in\{1,2,...,d\}$, $k\in\{1,2,...,f\}$, the parameters of  $\hat{X_i}$ and $\hat{X'_i}$ are identical. We know that for each $j\in\{1,2,...,d\}$, $k\in\{1,2,...,f\}$, $a_jPr(X_i=a_j|Y_i=a_k)\ge{0}$, thus $\text{Var}(\hat{X'}_i)$ is monotonically increasing with $f$. 
     
     Notice that the parameters of $\hat{X_i}$ and $\hat{X'_i}$ can not be identical as for at least one $j$, $q^i_{kj}$ will increase for $k\in\{f+1,f+2,...,d\}$, $j\in\{1,2,...,f\}$. However, this will make each $Pr(X_i=a_k|Y_i=a_j)$ smaller, thus  $Pr(X_i=a_k|Y_i=a_j)>Pr(X'_i=a_k|Y'_i=a_j)$.
     
     As a result: $\text{Var}(X_i)>\text{Var}(X'_i)$.
     
\textbf{When $d<f$,}  this case can be viewed as a special case of the general model with $P^i_{d+1}=P^i_{d+2}=...=P^i_{f}=0$. Thus the optimal solutions is straightforward: $q^i_{kk}=1-(1-P^i_k)/e^{\epsilon}$, $q^i_{kj}=P^i_j/e^{\epsilon}$ for $k,j\in\{1,2,...,d\}$; $q^i_{kj}=0$, for $k\in\{1,2,...,d\}$, $j\in\{d+1,d+2,...,f\}$. As a result, the optimal solution is equivalent to the case of the general model with $d=f$.
 In summary, the optimal range of output is $f=d$.
 \end{proof}

\section{Proof of Theorem 4}\label{prf:unary}
The privacy constraints can be expressed as:
\begin{small}
\begin{equation}\label{eq:unary2}
\begin{aligned}
&{\frac{Pr\left(\{Y^k_i\}_{k=1}^{|\mathcal{R}|}=y_1^{|\mathcal{R}|}\right)}{Pr\left(\{Y^k_i\}_{k=1}^{|\mathcal{R}|}=y_1^{|\mathcal{R}|}|\{U^k_i\}_{k=1}^{|\mathcal{R}|}=u_1^{|\mathcal{R}|}\right)}}\\
=&\frac{\sum_{\bar{u}_1^{|\mathcal{R}|}\in{\mathcal{B}^{|\mathcal{R}|}}}Pr(\{U^k_i\}_{k=1}^{|\mathcal{R}|}=\bar{u}_1^{|\mathcal{R}|})\prod_{k=1}^{|\mathcal{R}|}Pr(Y^k_i=y_k|U^k_i=\bar{u}_k\})}{\prod_{k=1}^{|\mathcal{R}|}Pr(Y^k_i=y_k|U^k_i=u_k))}\\
=&Pr\left(\{U^k_i\}_{k=1}^{|\mathcal{R}|}={u}_1^{|\mathcal{R}|}\right)\\
+&\frac{\sum_{\bar{u}_1^{|\mathcal{R}|}\neq{u_1^{|\mathcal{R}|}}}Pr(\{U^k_i\}_{k=1}^{|\mathcal{R}|}=\bar{u}_1^{|\mathcal{R}|})\prod_{k=1}^{|\mathcal{R}|}Pr(Y^k_i=y_k|U^k_i=\bar{u}_k)}{\prod_{k=1}^{|\mathcal{R}|}Pr(Y^k_i=y_k|U^k_i=u_k))}.
\end{aligned}
\end{equation}
\end{small}
Note that for any given output vector $y_1^k$, the product of $\prod_{k=1}^{|\mathcal{R}|}Pr(Y^k_i=y_k|U^k_i=\bar{u}_k\})$ and $\prod_{k=1}^{|\mathcal{R}|}Pr(Y^k_i=y_k|U^k_i=u_k))$ differ in at most two bits, because different values of $R$ results in only two bits difference when transferred into vector. To this end, the privacy metric of Eq. \eqref{eq:unary2}
is bounded by:
\begin{small}
\begin{equation}\label{boound_unary}
\begin{aligned}
\Bigg[P^i_r+&\frac{(1-P^i_r)q^i_{01}q^i_{10}}{(1-q^i_{01})(1-q^i_{10})}, P^i_r+\frac{(1-P^i_r)(1-q^i_{01})(1-q^i_{10})}{q^i_{01}q^i_{10}}\Bigg],
\end{aligned}
\end{equation}
\end{small}
As Eq. \eqref{boound_unary} must fall in the region of $[e^{-\epsilon},e^{\epsilon}]$ for all $r\in{\mathcal{R}}$, we have:
\begin{small}
\begin{equation}
\begin{aligned}
&P^i_{\min}+(1-P^i_{\min})\frac{q^i_{01}q^i_{10}}{(1-q^i_{01})(1-q^i_{10})}\ge{e^{-\epsilon}}\\
&P^i_{\min}+(1-P^i_{\min})\frac{(1-q^i_{01})(1-q^i_{10})}{q^i_{01}q^i_{10}}\ge{e^{\epsilon}},
\end{aligned}
\end{equation}
\end{small}
where $P^i_{\min}=\min_{r\in\mathcal{R}}P^i_r$. Then, the upper bound of the ratio of $\frac{(1-q^i_{01})(1-q^i_{10})}{q^i_{01}q^i_{10}}$ becomes (when $e^{-\epsilon}-P^i_{\min}\ge{0}$):
\begin{equation}\label{upperbound_unary}
    \frac{(1-q^i_{01})(1-q^i_{10})}{q^i_{01}q^i_{10}}\le{\frac{e^{\epsilon}-P^i_{\min}}{1-P^i_{\min}}}.
\end{equation}
The privacy constraints are just met when the inequality in Eq. \eqref{upperbound_unary} becomes equality. 
Note that there are more $0$s than $1$ in any input vector $\{U_i^k\}_{k=1}^{|\mathcal{R}|}$, and the utility function of
\begin{equation}\label{utility-Unary}
\begin{aligned}
    &E\left[\left(\{U_i^k\}_{k=1}^{|\mathcal{R}|}-E\left[\{U_i^k\}_{k=1}^{|\mathcal{R}|}|\{Y_i^k\}_{k=1}^{|\mathcal{R}|}\right]\right)^2\right]\\
    =&\sum_{k=1}^{|\mathcal{R}|}\left\{\text{Var}[U_i^k]-\text{Var}\left[E[U_i^k|Y_i^k]\right]\right\},
    \end{aligned}
\end{equation}
is a linear combination of MSEs of all errors. Therefore, to minimize MSE, we first set $q^i_{01}$ to be as small as possible. As a result, $q^{i*}_{10}=\frac{1}{2}$ $q^{i*}_{01}=\frac{1-P^i_{\min}}{e^{\epsilon}-2P^i_{\min}+1}$.






\ifCLASSOPTIONcaptionsoff
  \newpage
\fi

\end{document}